\documentclass[nonacm]{acmart}
\AtBeginDocument{%
  \providecommand\BibTeX{{%
    \normalfont B\kern-0.5em{\scshape i\kern-0.25em b}\kern-0.8em\TeX}}}

\usepackage{graphicx}
\usepackage[utf8]{inputenc}
\usepackage{amsmath,amsfonts}
\usepackage{algorithm2e}
\usepackage{algorithmic}
\usepackage{multirow}
\usepackage{url}
\usepackage{color}
\usepackage{rotating}
\usepackage[misc]{ifsym}
\usepackage{microtype}

\newenvironment{proof*}[1]
  {\renewcommand{\proofname}{#1}%
   \begin{proof}}
  {\end{proof}}

\newcommand{\reals}{\mathbb{R}}

\DeclareMathOperator*{\argmin}{arg\,min}

\newcommand{\dVec}{\mathbf{d}}
\newcommand{\wVec}{\mathbf{w}}
\newcommand{\xVec}{\mathbf{x}}

\newcommand{\CMat}{\mathbf{C}}
\newcommand{\cVec}{\mathbf{c}}
\newcommand{\ones}{\mathbf{1}}

\newcommand{\cmax}{c_{\textrm{max}}}
\newcommand{\cmin}{c_{\textrm{min}}}
\newcommand{\Cabs}{\mathbf{C}_{\textrm{abs}}}
\newcommand{\dabs}{\mathbf{d}_{\textrm{abs}}}
\newcommand{\PATHATTACK}{\texttt{PATHATTACK}}
\newcommand{\GreedyPathCover}{\texttt{GreedyPathCover}}
\newcommand{\Eupper}{E_\mathrm{upper}}
\newcommand{\bupper}{b_\mathrm{upper}}
\newcommand{\Elower}{E_\mathrm{lower}}
\newcommand{\blower}{b_\mathrm{lower}}
\newcommand{\Emid}{E_\mathrm{mid}}
\newcommand{\bmid}{b_\mathrm{mid}}
\newcommand{\Etemp}{E_\mathrm{temp}}
\newcommand{\btemp}{b_\mathrm{temp}}
\newcommand{\cost}{\mathrm{cost}}
\newcommand{\opt}{\mathrm{opt}}
\newcommand{\Enew}{E_\mathrm{new}}
\newcommand{\approxFactor}{\mathrm{apxFactor}}
\newcommand{\ctr}{\mathrm{ctr}}
\newcommand{\instA}{\mathcal{A}}
\newcommand{\instB}{\mathcal{B}}
\newcommand{\Ekeep}{E_\mathrm{keep}}
\newcommand{\Ealways}{E_\mathrm{always}}
\newcommand{\Enever}{E_\mathrm{never}}
\newcommand{\Ebest}{E_\mathrm{best}}
\newcommand{\cbest}{c_\mathrm{best}}

\newcommand{\hide}[1]{}

\begin{document}

\title{Attacking Shortest Paths by Cutting Edges}

\author{Benjamin A. Miller}
\email{miller.be@northeastern.edu}
\orcid{0000-0002-1649-1401}
\affiliation{%
  \institution{Northeastern University}
  \city{Boston}
  \state{MA}
  \country{USA}
  \postcode{02115}
}

\author{Zohair Shafi}
\email{shafi.z@northeastern.edu}
\orcid{0000-0001-6154-1466}
\affiliation{%
  \institution{Northeastern University}
  \city{Boston}
  \state{MA}
  \country{USA}
  \postcode{02115}}

\author{Wheeler Ruml}
\affiliation{%
  \institution{University of New Hampshire}
  \city{Durham}
  \state{NH}
  \country{USA}
  \postcode{03824}
}
\email{ruml@cs.unh.edu}
\orcid{0000-0002-1308-2311}

\author{Yevgeniy Vorobeychik}
\affiliation{%
 \institution{Washington University in St. Louis}
 \city{St. Louis}
 \state{MO}
 \country{USA}
 \postcode{63130}}
\email{yvorobeychik@wustl.edu}
\orcid{0000-0003-2471-5345}

\author{Tina Eliassi-Rad}
\affiliation{%
  \institution{Northeastern University}
  \city{Boston}
  \state{MA}
  \country{USA}
  \postcode{02115}}
\email{t.eliassirad@northeastern.edu}
\orcid{0000-0002-1892-1188}

\author{Scott Alfeld}
\affiliation{%
  \institution{Amherst College}
  \city{Amherst}
  \state{MA}
  \country{USA}
  \postcode{01002}}
\email{salfeld@amherst.edu}
\orcid{0000-0001-8446-4993}

\renewcommand{\shortauthors}{Miller et al.}

\begin{abstract}
Identifying shortest paths between nodes in a network is a common graph analysis problem that is important for many applications involving routing of resources. An adversary that can manipulate the graph structure could alter traffic patterns to gain some benefit (e.g., make more money by directing traffic to a toll road). This paper presents the \emph{Force Path Cut} problem, in which an adversary removes edges from a graph to make a particular path the shortest between its terminal nodes. We prove that this problem is APX-hard, but introduce \texttt{PATHATTACK}, a polynomial-time approximation algorithm that guarantees a solution within a logarithmic factor of the optimal value. In addition, we introduce the \emph{Force Edge Cut} and \emph{Force Node Cut} problems, in which the adversary targets a particular edge or node, respectively, rather than an entire path. We derive a nonconvex optimization formulation for these problems, and derive a heuristic algorithm that uses \texttt{PATHATTACK} as a subroutine. We demonstrate all of these algorithms on a diverse set of real and synthetic networks, illustrating the network types that benefit most from the proposed algorithms.
\end{abstract}


\maketitle

\section{Introduction}
\label{sec:intro}
Finding shortest paths among interconnected entities is an important task in a wide variety of applications. When routing resources---such as traffic on roads, ships among ports, or packets among routers---identifying the shortest path between two nodes is key to making efficient use of the network. 
Given that traffic prefers to take the shortest route, a malicious adversary with the ability to alter the graph topology could manipulate the paths to gain an advantage, e.g., direct traffic between two popular locations across a toll road the adversary owns. 
Countering such behavior is important, and understanding vulnerability to such manipulation is a step toward more robust graph mining.

In this paper, we present the \emph{Force Path Cut} problem in which an adversary wants the shortest path between a source node and a target node in an edge-weighted network to follow a preferred path. The adversary achieves this goal by cutting edges, each of which has a known cost for removal. We show that this problem is APX-hard via a reduction from the 3-Terminal Cut problem~\cite{Dahlhaus1994}. To solve Force Path Cut, we recast it as a Weighed Set Cover problem, which allows us to use well-established approximation algorithms to minimize the total edge removal cost. We propose the \texttt{PATHATTACK} algorithm, which combines these algorithms with a constraint generation method to efficiently identify paths to target for removal. While these algorithms only guarantee an approximately optimal solution in general, \texttt{PATHATTACK} yields the lowest-cost solution in a large majority of our experiments.

We also introduce the \emph{Force Edge Cut} and \emph{Force Node Cut} problems, where a specific edge (or node) is targeted rather than an entire path. These problems are also APX-hard, and we use \PATHATTACK{} as part of a heuristic search algorithm to solve them. The three problems are defined formally in the following section.

\subsection{Problem Statement}
\label{subsec:problem}
We consider a graph $G=(V, E)$, where the vertex set $V$ contains $N$ entities and $E$ consists of $M$ edges, which may be directed or undirected. Each edge has a weight $w:E\rightarrow\reals_{\geq0}$ denoting the expense of traversal (e.g., distance or time). In addition, each edge has a removal cost $c:E\rightarrow\reals_{\geq0}$. We are also given a source node $s\in V$, a target node $t\in V$, and a budget $b > 0$ for edge removal. Within this context, there are three problems we address:
\begin{itemize}
    \item {\bf Force Edge Cut}: Given an edge $e^*\in E$, find $E^\prime\subset E$ where $\sum_{e\in E^\prime}{c(e)}\leq b$ and all shortest paths from $s$ to $t$ in $G^\prime=(V, E\setminus E^\prime)$ use $e^*$.
    \item {\bf Force Node Cut}: Given a node $v^*\in V$, find $E^\prime\subset E$ where $\sum_{e\in E^\prime}{c(e)}\leq b$ and all shortest paths from $s$ to $t$ in $G^\prime=(V, E\setminus E^\prime)$ use $v^*$.
    \item {\bf Force Path Cut}: Given a path $p^*$ from $s$ to $t$ in $G$, find $E^\prime\subset E$ where $\sum_{e\in E^\prime}{c(e)}\leq b$ and $p^*$ is the shortest path from $s$ to $t$ in $G^\prime=(V, E\setminus E^\prime)$.
\end{itemize}
Each variation of the problem addresses a different adversarial objective: there is a particular edge (e.g., a toll road), a particular node (e.g., a router in a network), or an entire path (e.g., a sequence of surveillance points) where increased traffic would benefit the adversary. The attack vector in all cases is removal of edges, and the adversary has access to the entire graph.
\subsection{Contributions}
The main contributions of this paper as as follows:
\begin{itemize}
    \item We define the Force Edge Cut, Force Node Cut, and Force Path Cut problems and prove that they are APX-hard.
    \item We introduce the PATHATTACK algorithm, which provides a logarithmic approximation for Force Path Cut with high probability in polynomial time.
    \item We provide a non-convex optimization formulation for Force Edge Cut and Force Node Cut, as well as polynomial-time heuristic algorithms.
    \item We present the results of over 16,000 experiments on a variety of synthetic and real networks, demonstrating where these algorithms perform best with respect to baseline methods.
\end{itemize}

\subsection{Paper Organization}
The remainder of this paper is organized as follows. In Section~\ref{sec:related}, we briefly summarize related work on inverse optimization and adversarial graph analysis. In Section~\ref{sec:complexity}, we provide a sketch of the proof that all three problems are APX-hard. Section~\ref{sec:PATHATTACK} introduces the \PATHATTACK{} algorithm, which guarantees a logarithmic approximation of the optimal solution to Force Path Cut. In Section~\ref{sec:edgeAlgorithm}, we show how \PATHATTACK{} can be used as a heuristic to solve Force Edge Cut and Force Node Cut. Section~\ref{sec:experiments} documents the results of experiments on diverse real and synthetic networks, demonstrating where \PATHATTACK{} provides a benefit over baseline methods. In Section~\ref{sec:conclusion}, we conclude with a summary and discussion of open problems.

\section{Related Work}
\label{sec:related}

Early work on attacking networks focused on disconnecting them~\cite{Albert2000}. This work demonstrated that targeted removal of high-degree nodes was highly effective against networks with power law degree distributions (e.g., Barab\'{a}si--Albert networks), but far less so against random networks. This is due to the prevalence of hubs in networks with such degree distributions. Other work has focused on disrupting shortest paths via edge removal, but in a narrower context than ours. Work on the most vital edge problem (e.g., \cite{Nardelli2003}) attempts to efficiently find the single edge whose removal most increases the distance between two nodes. Our present work, in contrast, considers a devious adversary that wishes a certain path to be shortest.

There are several other adversarial contexts in which path-finding is highly relevant. Some work is focused on traversing hostile territory, such as surreptitiously planning the path of an unmanned aerial vehicle~\cite{Jun2003}. The complement of this is network interdiction, where the goal is to intercept an adversary who is attempting to traverse the graph while remaining hidden. Bertsimas et al. formulate an optimization problem similar to Force Path Cut, where the goal is overall disruption of flows rather than targeting a specific shortest path~\cite{Bertsimas2016}. Network interdiction has been studied in a game theoretic context for many years~\cite{Washburn1995}, and has expanded into work on disrupting attacks, with the graph representing an attack plan~\cite{Letchford2013}. In this work, as in ours, oracles can be used to avoid enumerating an exponentially large number of possible strategies~\cite{Jain2011}.

As with many graph problems, finding a shortest path can be formulated as an optimization problem: over all paths from $s$ to $t$, find the one that minimizes the sum of edge weights. Finding the weights for two nodes to have a particular shortest path is an example of inverse optimization~\cite{Ahuja2001}. From this perspective, the shortest path is a function mapping edge weights to a path, and the inverse shortest path problem is to find the inverse function: input a path and output a set of weights. This typically involves finding a new set of weights given a baseline that should change as little as possible, with respect to some distance metric (though it is shown in~\cite{Xu1995} that solving the problem without baseline weights solves the minimum cut problem). To solve inverse shortest path while minimizing the $L_1$ norm between the weights, Zhang et al. propose a column generation algorithm~\cite{Zhang1995}, similar to the constraint generation method we describe in Section~\ref{subsec:constraintGeneration}. Such a constraint generation procedure
is also used in the context of navigation meshes in~\cite{Brandao2021}. The instance of inverse shortest path that is closest to Force Path Cut uses the weighted Hamming distance between the weight vectors, where changing the weight of an edge has a fixed cost, regardless of the size of the change.\footnote{Edge removal can be simulated by significantly increasing weights. Inverse shortest path using weighted Hamming distance also allows for reducing edge weights, which is not allowed in Force Path Cut.} This case was previously shown to be NP-hard~\cite{Zhang2008}. In this paper, we show that Force Path Cut, which is less flexible, is also not just NP-hard, but APX-hard.

In addition to inverse shortest path, authors have considered the inverse shortest path length problem (sometimes called reverse shortest path)~\cite{Tayyebi2016}, where only the length of the shortest path is specified. This problem has been shown to be NP-hard except when only considering a single source and destination~\cite{Zhang2003,Cui2010}. There is also the notion of inverse shortest path routing, in which edge weights are manipulated so that edges from one specified subset are used for shortest paths, while edges from a second, disjoint subset are never used~\cite{Call2011}. Other graph-based optimization problems, such as maximum flow and minimum cut, have also been topics in the inverse optimization literature (see, e.g., \cite{Liu2006,Jiang2010,Deaconu2020}).

There has recently been a great deal of work on attacking machine learning methods where graphs are part of the input. Finding shortest paths is another common graph problem that attackers could exploit. Attacks against vertex classification~\cite{Dai2018,Zugner2018,Chen2018,WuH2019,Xu2019} and node embeddings~\cite{Bojchevski2019, Chang2020} consider attackers that can manipulate edges, node attributes, or both in order to affect the outcome of the learning method. In addition, attacks against community detection have been proposed where a node can create new edges to alter its group assignment from a community detection algorithm~\cite{Kegelmeyer2018,Chen2020a,Jiang2020}, or to reduce the efficacy of community detection overall~\cite{Nagaraja2010,Waniek2018,Fionda2017,LiJ2020,Chen2019}. Our work complements these efforts, expanding the space of adversarial graph analysis into another important graph mining task.
\section{Computational Complexity}
\label{sec:complexity}
We first consider the complexity of Force Path Cut. Like inverse shortest path under Hamming distance~\cite{Zhang2008}, this problem is NP-hard, as shown in our prior work~\cite{Miller2021}. A novel result of our present work is that Force Path Cut is not merely NP-hard; there is no possible polynomial time approximation scheme, i.e., no polynomial-time algorithm can find a solution within a factor of $(1+\epsilon)$ of the optimal (minimal) budget unless P=NP. Using a linear reduction from another APX-hard problem, we prove the following theorem.
\begin{theorem}
\label{thm:pathHard}
Force Path Cut is APX-hard, including the case where all weights and all costs are equal.
\end{theorem}
\subsection{Proof Sketch}
\label{subsec:sketch}
The following is a sketch of the proof of Theorem~\ref{thm:pathHard}, with a detailed proof provided in Appendix~\ref{sec:proof}. First, consider the case for undirected graphs, where all edges weights are equal to their removal costs. We prove this is APX-hard via reduction from 3-Terminal Cut~\cite{Dahlhaus1994}. In 3-Terminal Cut, we are given a graph $G=(V, E)$, three terminal nodes $s_1,s_2,s_3\in V$, and a budget $b\geq 0$. For the purpose of this proof, we only consider the case where all edge weights are equal. The goal is to determine whether there exists $E^\prime\subset E$, with $|E^\prime|\leq b$, where $s_1$, $s_2$, and $s_3$ are disconnected in $G^\prime=(V, E\setminus E^\prime)$, i.e., after the edges is $E^\prime$ are removed, no path exists between any of the terminal nodes.

We reduce 3-Terminal Cut to Force Path Cut as follows. Create $M+1$ new disjoint paths, each $N$ hops long, from $s_1$ to $s_2$. This introduces $(M+1)(N-1)$ new nodes and $(M+1)N$ new edges. Do the same for $s_2$ and $s_3$ (another $(M+1)(N-1)$ nodes and $(M+1)N$ edges). Create a new $(2N-1)$-hop path from $s_1$ to $s_3$ ($2N-2$ new nodes, $2N-2$ new edges). Make the new path from $s_1$ to $s_3$ the target path $p^*$, with $s=s_1$ and $t=s_3$. Call the resulting graph $\hat{G}$. Solve Force Path Cut on $\hat{G}$, obtaining a set of edges $E^\prime$ to be removed. This process is illustrated in Fig.~\ref{fig:reduction}.

If $E^\prime$ is a solution to Force Path Cut, it is also a solution to 3-Terminal Cut. If any of the terminal nodes could be reached from another using edges in the original graph, $p^*$ would not be the shortest path. If a path from $s_2$ to $s_3$ from the original graph remained, for example, this path would be at most $N-1$ hops long. There would exist a path from $s_1$ to $s_3$ using one of the new $(N+1)$-hop paths from $s_1$ to $s_2$, followed by the remaining path from $s_2$ to $s_3$. This path would be less than $2N+1$ hops, so $p^*$ would not be the shortest path.
\begin{figure}
    \centering
    \includegraphics[height=1.4in]{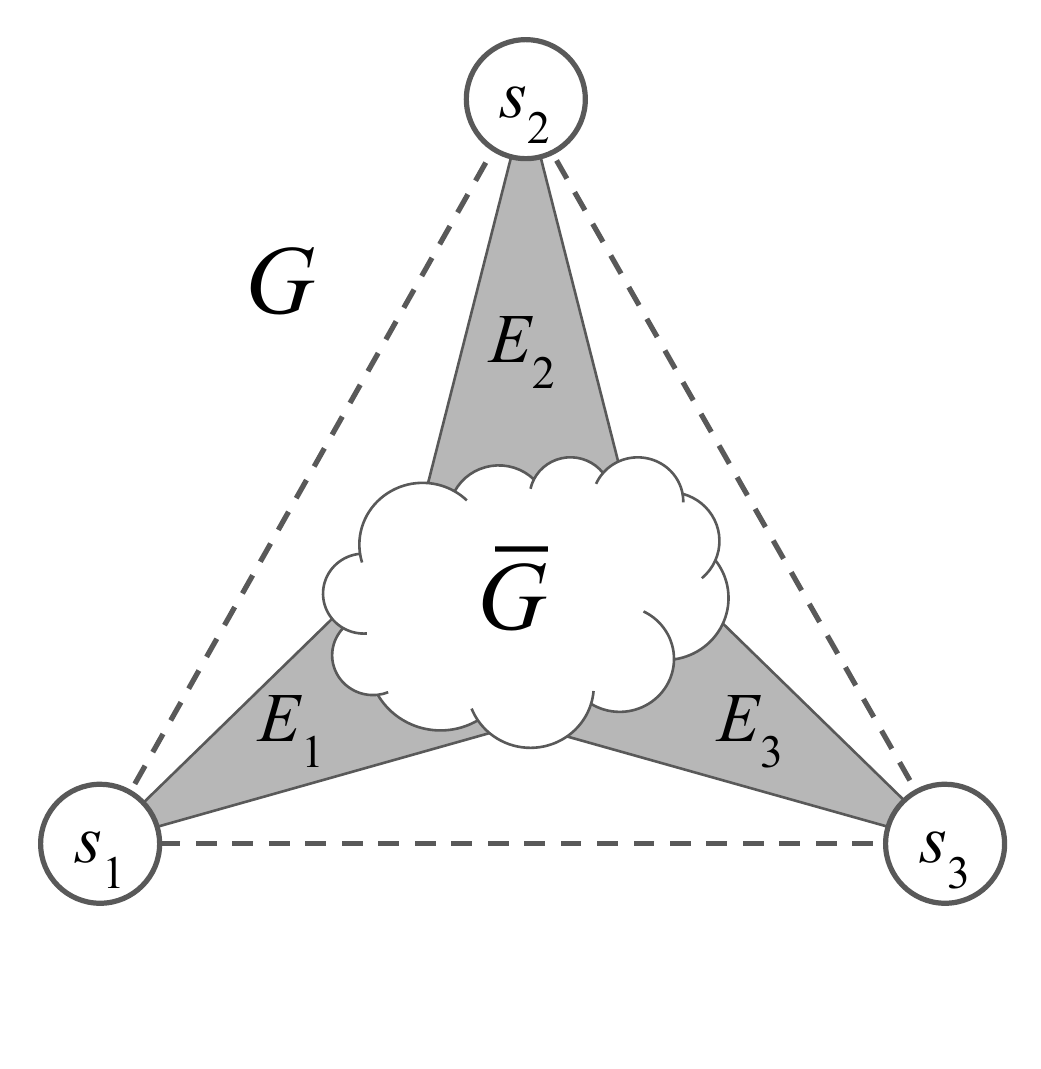}
    \includegraphics[height=1.4in]{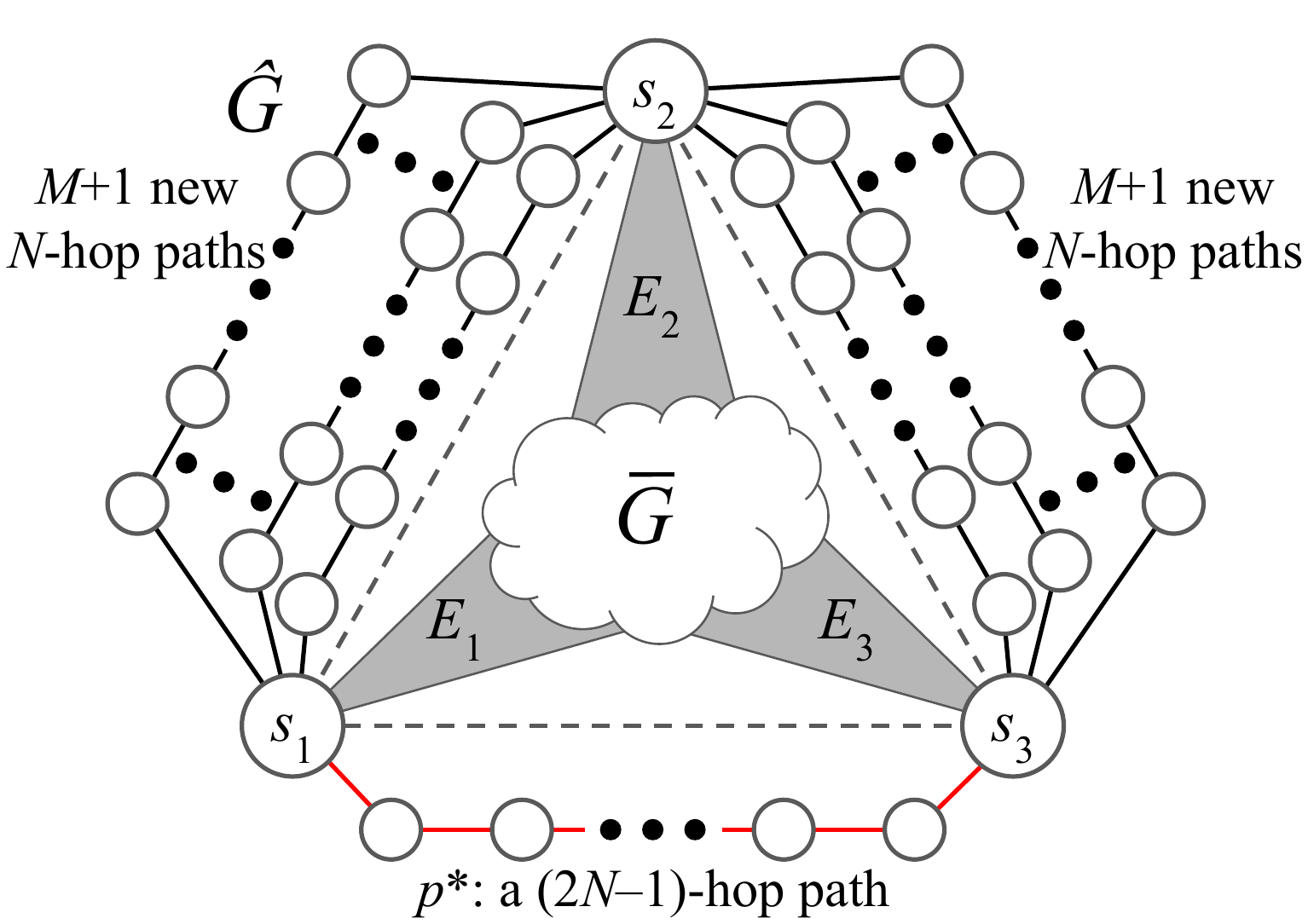}
    \includegraphics[height=1.4in]{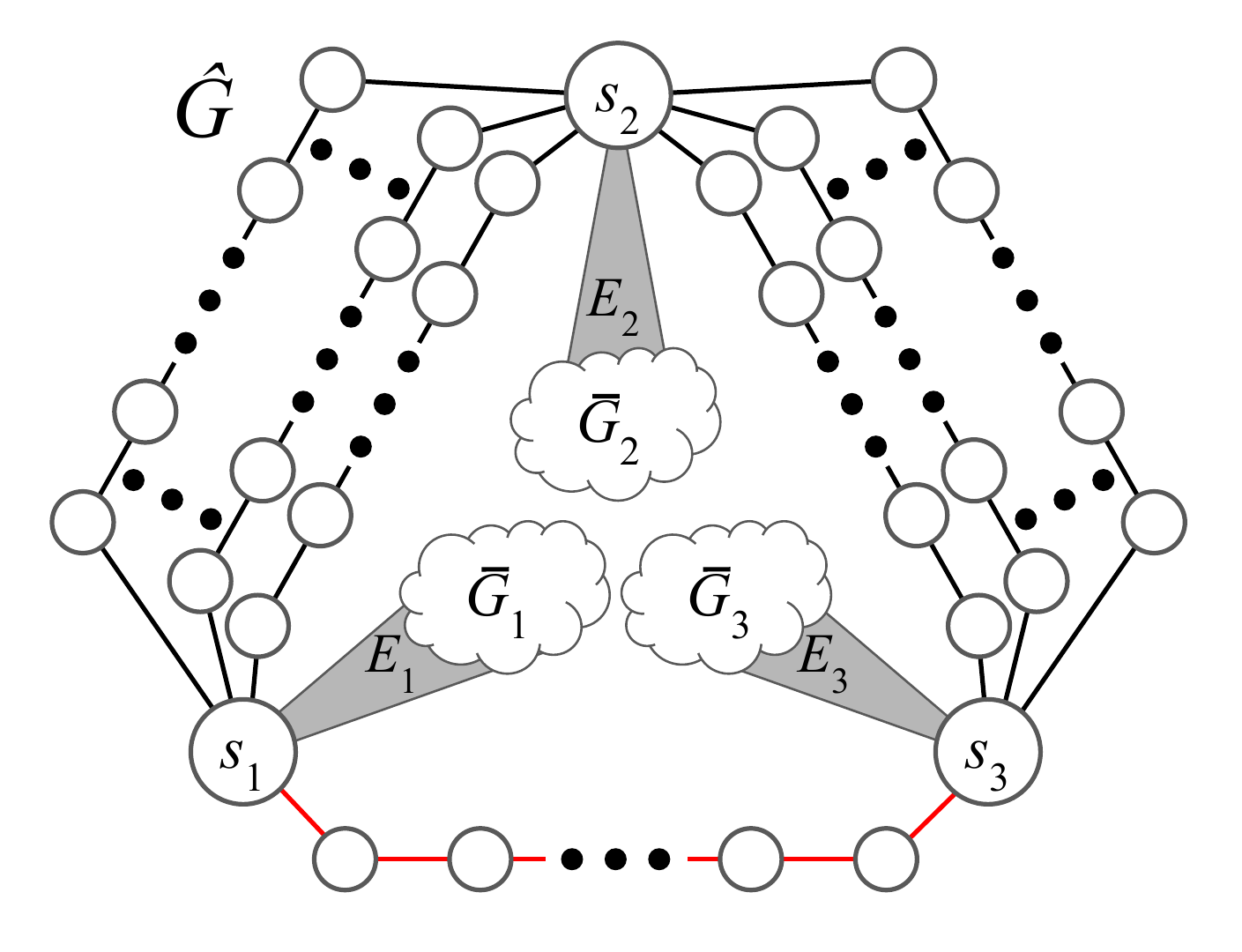}
    \caption{Reduction from 3-Terminal Cut to Force Path Cut. The initial graph (left) includes 3 terminal nodes $s_1$, $s_2$, and $s_3$, which are connected to the rest of the graph by edge sets $E_1$, $E_2$, and $E_3$, respectively. The dashed lines indicate the possibility of edges between terminals. The input to Force Path Cut, $\hat{G}$ (center), includes the original graph plus many long parallel paths between terminals. A single path from $s_1$ to $s_2$ comprising $p^*$ is indicated in red. The result of solving Force Path Cut and removing the returned edges (right) is that any existing paths between the terminals have been removed, thus disconnecting them in the original graph and solving 3-Terminal Cut.}
    \label{fig:reduction}
\end{figure}

In the optimization version of the problem, the goal is to find the $E^\prime$ that solves 3-Terminal Cut with the smallest budget. Let $E^\prime_{\mathrm{OPT}}$ be this edge set. Since this set disconnects the terminals, it also solves Force Path Cut. In the full proof, we show that there is no way to solve Force Path Cut on $\hat{G}$ with fewer edges. Thus, the optimal solution for Force Path Cut is the same as the optimal solution to 3-Terminal Cut. In addition, if we consider \emph{any} solution $\hat{E}^\prime$ to Force Path Cut on $\hat{G}$, we show that there is a solution to 3-Terminal Cut that is no larger than $\hat{E}^\prime$. As we discuss in the appendix, this means that the reduction is linear, and thus preserves approximation. Since 3-Terminal Cut is APX-hard, a linear reduction implies that Force Path Cut is also APX-hard.

To prove that Force Path Cut is APX-hard for directed graphs as well, we reduce Force Path Cut for undirected graphs to the same problem for directed graphs. Given an undirected graph $G=(V,E)$ and the path $p^*$ from source node $s$ to destination node $t$, create a new graph $\hat{G}=(V, \hat{E})$ on the same vertex set, with two directed edges for each undirected edge from $G$, i.e., $\hat{E}$ contains the directed edge $(u,v)$ if and only if $E$ contains either $\{u,v\}$. We again consider the case where all weights and all costs are equal. Solve Force Path Cut on $\hat{G}$, obtaining the optimal solution $\hat{E}^\prime$. For each directed edge in $(u, v)\in\hat{E}^\prime$, remove the undirected edge $\{u,v\}$ from $G$. In the resulting graph, $p^*$ will be the shortest path from $s$ to $t$. As we show in Appendix~\ref{subsec:proofDirected}, if the optimal solution includes $(u,v)$, it cannot also include $(v,u)$, i.e., if $(u,v)\in\hat{E}^\prime$, then $(v,u)\notin\hat{E}^\prime$. This implies that the optimal solution size for Force Path Cut in directed graphs is equal to the optimal solution size for Force Path Cut in undirected graphs. By the same argument that we used for undirected graphs, Force Path Cut is therefore APX-hard for directed graphs.

\subsection{Extension to Force Edge Cut and Force Node Cut}
The same reduction used to prove that Force Path Cut is APX-hard can show the same for Force Edge Cut. The difference is that, rather than letting $p^*$ be the new path from $s_1$ to $s_3$, we let $e^*$ be one of the edges along this new path. This will still force the new path from $s_1$ to $s_3$ to be the shortest, so the optimal will be exactly the same. Like Force Path Cut, the directed and undirected cases have the same solution, resulting in the following corollary.
\begin{corollary}
Force Edge Cut is APX-hard, including the case where all weights and all costs are equal.
\end{corollary}
By selecting a node along the new path from $s_1$ to $s_3$ to act as $v^*$, we can make the same argument for Force Node Cut.
\begin{corollary}
Force Node Cut is APX-hard, including the case when all weights and all costs are equal.
\end{corollary}

\section{PATHATTACK}
\label{sec:PATHATTACK}
While Force Path Cut is APX-hard, the optimal solution can be approximated within a logarithmic factor in polynomial time. This section describes an algorithm to achieve this.

\subsection{Cutting Paths Via Set Cover}
We first note that Force Path Cut can be recast as an instance of Weighted Set Cover. In weighted set cover, there is a discrete universe $U$ of elements and a set $\mathcal{S}$ of subsets of $U$, where each subset $S\in\mathcal{S}$ has an associated cost $c(S)\geq 0$. The objective is to find the subset of $\mathcal{S}$ with the lowest total cost where the union of the elements comprises $U$, i.e., to find
\begin{align}
    \hat{\mathcal{S}}=&\argmin_{\mathcal{S}^\prime\subseteq\mathcal{S}}{\sum_{S\in\mathcal{S}^\prime}c(S)}\\
    \text{s.t.}&~\bigcup_{S\in\mathcal{S}^\prime}{S}=U.
\end{align}
When we solve Force Path Cut, the goal is to cut all paths from $s$ to $t$ that are not longer than $p^*$, and to do so by selecting edges for removal. This is a set cover problem. The universe $U$ is the set of all paths competing to be shortest. This set must be ``covered'' by removing all such paths from the graph. If any of these paths remains, $p^*$ will not be the shortest path. The subsets in $\mathcal{S}$ are edges: Each edge is the set of paths that use that edge. Cutting the edge is selecting that set for inclusion in the union, i.e., removing all paths that include the edge. Figure~\ref{fig:translation} illustrates the connection between the two problems.
\begin{figure*}
    \centering
    \includegraphics[width=0.85\textwidth]{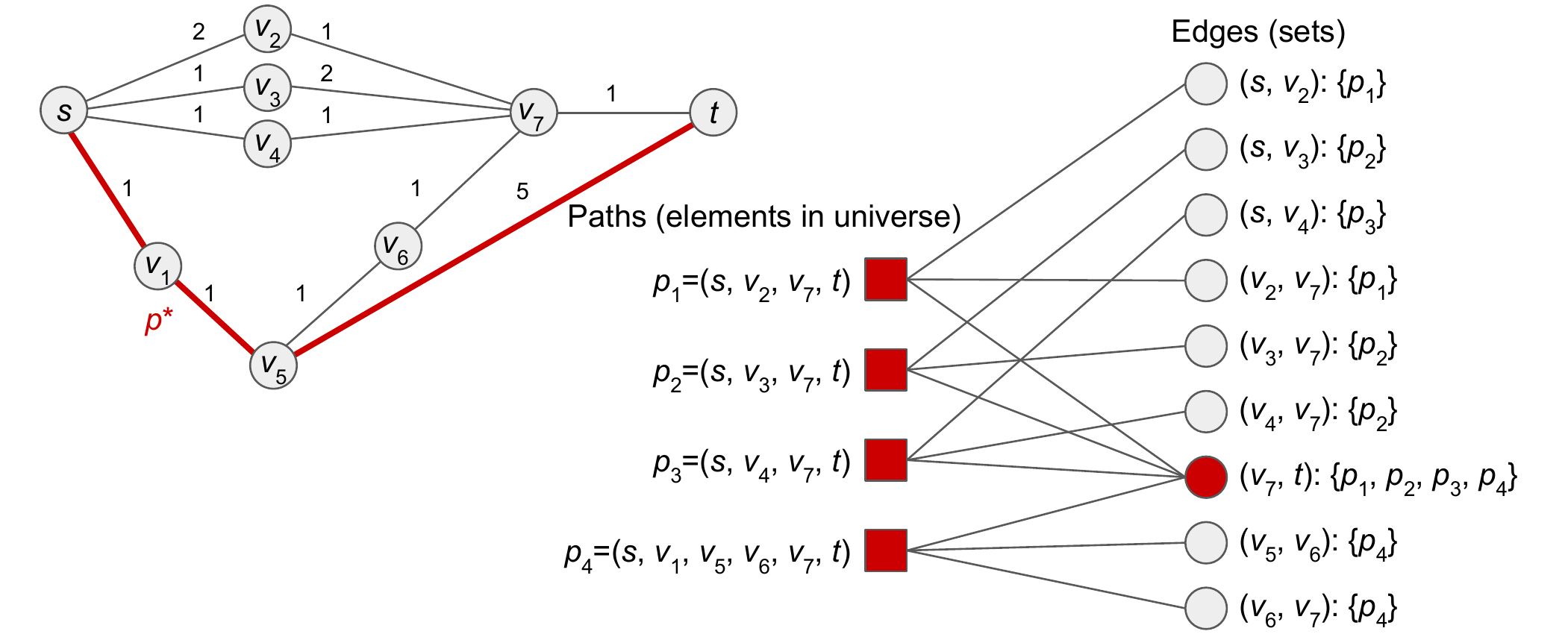}
    \caption{The Force Path Cut problem is an example of the Weighted Set Cover problem. In the bipartite graph on the right, the square nodes represent paths and the circle nodes represent edges. Note that edges along $p^*$ are not included. When the red-colored circle (i.e., edge $(v_7,t)$) is removed, then the red-colored squares (i.e., paths $p_1$, $p_2$, $p_3$, and $p_4$) are removed.}
    \label{fig:translation}
\end{figure*}

Weighted Set Cover is not only APX-hard, but cannot be approximated within less than a logarithmic factor~\cite{Raz1997}. There are, however, well known approximation algorithms that enable us to efficiently reach this asymptotic lower bound. The challenge for Force Path Cut is that the universe of competing paths may be extremely large. We address this challenge in the remainder of this section.

\subsection{Approximation Algorithms for Fixed Path Sets}
We consider two approximation algorithms for Weighted Set Cover, both of which achieve a logarithmic factor approximation for the optimal solution~\cite{Vazirani2003}. The first is a greedy method, which iteratively selects the most cost effective set and includes it in the solution, i.e., it selects the set where the number of uncovered elements divided by the cost is the lowest. We likewise iteratively remove the edge that cuts the most uncut paths in $P$ for the lowest cost. We refer to this procedure as \GreedyPathCover{}, and we provide its pseudocode in Algorithm~\ref{alg:greedySC}. We have a fixed set of paths $P\subset P_{p^*}$. Note that this algorithm only uses costs, not weights: the paths of interest have already been determined and we only need to determine the cost of breaking them. \texttt{GreedyPathCover} performs a constant amount of work at each edge in each path in the initialization loop and the edge and path removal. We use lazy initialization to avoid initializing entries in the tables associated with edges that do not appear in any paths. Thus, populating the tables and removing paths takes time that is linear in the sum of the number of edges over all paths, which in the worst case is $O(|P|N)$. Finding the most cost-effective edge takes $O(M)$ time with a na\"{i}ve implementation, and this portion is run at most once per path, leading to an overall running time of $O(|P|(N+M))$. Using a more sophisticated data structure, like a Fibonacci heap, to hold the number of paths for each edge would enable finding the most cost effective edge in constant time, but updating the counts when edges are removed would take $O(\log{M})$ time, for an overall running time of $O(|P|N\log{M})$. The worst-case approximation factor is the harmonic function of the size of the universe~\cite{Vazirani2003}, i.e., $H_{|\mathcal{U}|}=\sum_{n=1}^{|\mathcal{U}|}{1/n}$,
which implies that the \texttt{GreedyPathCover} algorithm has a worst-case approximation factor of $H_{|P|}$. 
\begin{algorithm}
\KwIn{Graph $G=(V, E)$, costs $c$, target path $p^*$, path set $P$, edges $\Ekeep$}
\KwOut{Set $E^\prime$ of edges to cut}
$T_P\gets$ empty hash table;  \tcp{set of paths for each edge}
$T_E\gets$ empty hash table; \tcp{set of edges for each path}
$N_P\gets$ empty hash table; \tcp{path count for each edge}
\ForEach{$e \in E$}{
$T_P[e]\gets\emptyset$\;
$N_P[e]\gets0$\;
}
\ForEach{$p\in P$}{
$T_E[p]\gets\emptyset$\;
\ForEach{$e\in E_p\setminus \Ekeep$}{
$T_P[e]\gets T_P[e]\cup \{p\}$\;
$T_E[p]\gets T_E[p]\cup \{e\}$\;
$N_P[e]\gets N_P[e]+1$\;
}
}
$E^\prime\gets\emptyset$\;
\While{$\max_{e\in E}{N_P[e]}>0$}{
$e^\prime\gets\arg\max_{e\in E}{N_P[e]/c(e)}$; \tcp{find most cost-effective edge}
$E^\prime\gets E^\prime\cup\{e^\prime\}$\;
\ForEach{$p \in T_P[e^\prime]$}{
\ForEach{$e_1\in T_E[p]$}{
$N_P[e_1] \gets N_P[e_1] - 1$; \tcp{decrement path count}
$T_P[e_1] \gets T_P[e_1] \setminus \{p\}$; \tcp{remove path}
}
$T_E[p]\gets\emptyset$; \tcp{clear edges}
}
}
\Return $E^\prime$
\caption{\texttt{GreedyPathCover}}
\label{alg:greedySC}
\end{algorithm}

For the second approximation algorithm, we introduce the integer program formulation of Force Path Cut, modeled after the formulation for Weighted Set Cover. The objective is to minimize the cost of the edges that are cut. Let $\cVec\in\reals_{\geq0}^M$ be a vector of edge removal costs, where each entry corresponds to an edge. The binary vector $\Delta\in\{0,1\}^M$ indicates the edges to be removed. In the integer program formulation of Weighted Set Cover, each dimension corresponds to a set, and each constraint corresponds to an element. Each element forms a linear constraint, forcing at least one of the sets containing that element to be selected. To solve Force Path Cut, let $P_{p^*}$ be the set of paths that must be cut. For each path $p\in P_{p^*}$, let $\xVec_p\in \{0,1\}^M$ be the indicator vector for $E_p$, the set of edges along $p$. We solve Force Path Cut via the integer program
\begin{align}
    \hat{\Delta}=&\argmin_{\Delta}\cVec^\top\Delta\label{eq:minCostObj}\\
    \text{s.t.} &~\Delta\in\{0, 1\}^{M}\\
    &~\xVec_p^\top\Delta\geq1~~\forall p\in P_{p^*}\\
    &~\xVec_{\textrm{keep}}^\top\Delta=0\label{eq:dontCutP*}.
\end{align}
The constraint (\ref{eq:dontCutP*}) prevents any edges in the set $\Ekeep$ from being cut.\footnote{This could also be achieved by removing these variables from the optimization.} With target paths, we set $\Ekeep=E_{p^*}$, but we allow greater flexibility for the methods outlined in Section~\ref{sec:edgeAlgorithm}, which target nodes or edges.

The second algorithm is a randomized algorithm that uses this formulation. As with the greedy approach, we focus on a subset of paths $P\subset P_{p^*}$. This algorithm relaxes the integer program (\ref{eq:minCostObj})--(\ref{eq:dontCutP*}) by replacing the binary constraint $\Delta\in\{0,1\}^M$ with a continuous constraint $\Delta\in[0,1]^M$. When we find the optimal solution $\hat{\Delta}$ to the relaxed problem, we perform the following randomized rounding procedure for each edge $e$:
\begin{enumerate}
    \item Treat the corresponding entry $\hat{\Delta}_e$ as a probability.
    \item Draw $\lceil\ln{(4|P|)}\rceil$ i.i.d. Bernoulli random variables with probability $\hat{\Delta}_e$.\label{item:sample}
    \item Cut $e$ only if at least one random variable from step~\ref{item:sample}  is 1.
\end{enumerate}
The resulting graph must satisfy two criteria. First, the constraints must all be satisfied, i.e., all paths in $P$ must be cut. In addition, the cost of cutting the edges must not exceed $\ln{(4|P|)}$ times the fractional optimum, i.e., the optimal objective of the relaxed linear program (LP). If one of these criteria is not satisfied, the randomized rounding procedure is run again. Each criterion has less than $1/4$ probability of not being satisfied~\cite{Vazirani2003rounding}. Thus, the expected number of randomized rounding trials to achieve a satisfactory solution is less than 2. We present pseudocode for this procedure, which we call \texttt{RandPathCover}, in Algorithm~\ref{alg:LPCut}.
\begin{algorithm}
\KwIn{Graph $G=(V, E)$, costs $\cVec$, path $p^*$, path set $P$, edges $\Ekeep$}
\KwOut{Set $E^\prime$ of edges to cut}
$\hat{\Delta}\gets$ relaxed cut solution to (\ref{eq:minCostObj})--(\ref{eq:dontCutP*}) with paths $P$\;
$\Delta \gets \mathbf{0}$\;
$E^\prime\gets\emptyset$\;
not\_cut$\gets\mathbf{True}$\;
\While{$\cVec^\top\Delta > \cVec^\top\hat{\Delta}(4\ln{(4|P|)})$ $\mathbf{or}$ $\mathrm{not\_cut}$}{
$E^\prime\gets\emptyset$\;
\For{$i\gets 1$ to $\lceil\ln{(4|P|)}\rceil$}{
\tcp{randomly select edges based on $\hat{\Delta}$}
$E_1\gets\{e\in E \textrm{ with probability }\hat{\Delta}_e\}$\;
$E^\prime\gets E^\prime\cup E_1$\;
}
$\Delta\gets$ indicator vector for $E^\prime$\;
not\_cut$\gets(\exists p\in P$ where $E_p\cap E^\prime\neq \emptyset$)\;
}
\Return $E^\prime$
\caption{\texttt{RandPathCover}}
\label{alg:LPCut}
\end{algorithm}

\subsection{Constraint Generation}
\label{subsec:constraintGeneration}
Algorithms~\ref{alg:greedySC} and~\ref{alg:LPCut} work with a fixed set of paths $P$, but to solve Force Path Cut, enumeration of the full set of relevant paths may be intractable. Consider, for example, the case where $G$ is a clique (i.e., a complete graph) and all edges have weight 1 except the edge joining $s$ and $t$, which has weight $N$. Among all simple paths from $s$ to $t$, the longest path is the direct path $(s, t)$ that only uses one edge; all other simple paths are shorter, including $(N-2)!$ paths of length $N-1$. Setting $p^*=(s,t)$ makes the full set of constraints (i.e., one for each path that need to be cut) extremely large. However, if $P$ were to include only those paths of length $2$ and $3$, we would obtain the optimal solution with only $(N-2)^2+(N-2)$ constraints. In general, we can solve linear programming problems---even those with infinitely many constraints---as long as we have a polynomial-time method to identify a constraint that a candidate solution violates.

In our case, the oracle that provides a violated constraint, if one exists, is the shortest path algorithm. After applying a candidate solution, we see if $p^*$ is the shortest path from $s$ to $t$ in the resulting graph. If so, the procedure terminates. If not, the shortest path is added as a new constraint. This procedure works in an iterative process with the approximation algorithms to expand $P$ as new constraints are discovered. We refer to the resulting algorithm as \PATHATTACK{}, and provide pseudocode in Algorithm~\ref{alg:PATHATTACK}. Depending on whether we use \texttt{GreedyPathCover} or \texttt{RandPathCover}, we refer to the algorithm as \texttt{PATHATTACK-Greedy} or \texttt{PATHATTACK-Rand}, respectively.

\begin{algorithm}
\KwIn{Graph $G=(V, E)$, costs $c$, weights $w$, target path $p^*$, edges $E_\textrm{keep}$, cover alg. \texttt{Cvr}}
\KwOut{Set $E^\prime$ of edges to cut}
$E^\prime\gets\emptyset$\;
$P\gets\emptyset$\;
$\cVec\gets$ vector from costs $c(e)$ for $e\in E$\;
$G^\prime\gets(V, E\setminus E^\prime)$\;
$s, t\gets$ source and destination nodes of $p^*$\;
$p\gets$ shortest path from $s$ to $t$ in $G^\prime$ (not including $p^*$)\;
\While{$p$ is not longer than $p^*$}{
$P\gets P\cup\{p\}$\;
$E^\prime\gets$\texttt{Cvr}$(G, \cVec, p^*, P)$\;
$G^\prime\gets(V, E\setminus E^\prime)$\;
$p\gets$ shortest path from $s$ to $t$ in $G^\prime$ (not including $p^*$) using weights $w$\;
}
\Return $E^\prime$
\caption{\texttt{PATHATTACK}}
\label{alg:PATHATTACK}
\end{algorithm}
\subsection{\PATHATTACK{} Convergence and Approximation Guarantees}
While the approximation factor for Set Cover is a function of the size of the universe (all paths that need to be cut), this is not the fundamental factor in the approximation in our case. The approximation factor for \texttt{PATHATTACK-Greedy} is based only on the paths we consider explicitly. Using only a subset of constraints, the worst-case approximation factor is determined by the size of that subset. By the final iteration of \texttt{PATHATTACK}, however, we have a solution to Force Path Cut that is within $H_{|P|}$ of the optimum of the less constrained problem, using $|P|$ from the final iteration. 
This yields the following proposition:
\begin{proposition}
The approximation factor of \texttt{PATHATTACK-Greedy} is at most $H_{|P|}$ times the optimal solution to Force Path Cut.\label{prop:greedyFactor}
\end{proposition}

A similar argument holds for \texttt{PATHATTACK-Rand}, applying the results of \cite{Vazirani2003}:
\begin{proposition}
\texttt{PATHATTACK-Rand} yields a worst-case $O(\log |P|)$ approximation to Force~Path Cut with high probability.\label{prop:LPFactor}
\end{proposition}

There is also a variation of \texttt{PATHATTACK-Rand} can guarantee polynomial time convergence. While the number of implicit constraints may be extremely large, each one is a standard linear inequality constraint, which implies that the feasible region is convex. Given a polynomial-time oracle that returns a constraint violated by a proposed solution $\hat{\Delta}\in[0,1]^M$ to the relaxed LP, we could use Khachiyan's ellipsoid algorithm (see~\cite{Gacs1981}), which can solve a linear program with an arbitrary (or infinite) number of linear constraints in a polynomial number of iterations~\cite{Grotschel1981}. We use the randomized rounding procedure from Algorithm~\ref{alg:LPCut} to achieve such an oracle. In Appendix~\ref{sec:convergence}, we prove that this oracle terminates in polynomial time with high probability. This results in the following theorem.

\begin{theorem}
\PATHATTACK{} using the ellipsoid algorithm converges in polynomial time with high probability.
\label{thm:PATHATTACKconvergence}
\end{theorem}

While this guarantees polynomial time convergence, we do not use the ellipsoid algorithm in our experiments in Section~\ref{sec:experiments}, as other methods converge much faster in practice.

\subsection{PATHATTACK for Node Removal}
\label{subsec:nodeRemoval}
This paper is focused on the case where edges are removed from a graph, but there are applications in which removing nodes is more appropriate. In a computer network, for example, it may be more likely for a node to be taken offline (e.g., via a denial of service attack) rather than to have particular connections disallowed. Force Path Cut via node removal is a more complicated problem; for a given $p^*$, there is not always a solution. Consider, for example, a triangle of nodes $u$, $v$, and $w$, where all weights are equal, and $p^*=(u, v, w)$: none of the three nodes can be removed because it would destroy $p^*$, so $p^*$ cannot be made the shortest path due to the existence of the edge connecting $u$ to $w$.

In cases where it is possible, however, we can still use \PATHATTACK{}. The mapping to Weighted Set Cover is analogous: a path is an element and a node are sets that include all paths. Given a graph and $p^*$, we can check if it is possible to make $p^*$ the shortest path from $s$ to $t$ via node removal: if $p^*$ is the shortest path from $s$ to $t$ in the induced subgraph of the nodes on $p^*$, then there is a solution. If not, there is none; the method to check has demonstrated that, even if all other nodes are removed, $p^*$ is not the shortest path. If there is a possible solution, we can run \PATHATTACK{} for node removal and achieve the same convergence and approximation guarantees as we have for edge removal. As this is ancillary to the main focus of the paper, we relegate further discussion of node removal to Appendix~\ref{sec:nodeRemoval}.

\section{Heuristics for Force Edge Cut and Force Node Cut}
\label{sec:edgeAlgorithm}
Solving Force Edge Cut or Force Node Cut requires an additional layer of optimization. Since these problems do not consider a fixed path, there is the possibility that a candidate solution could be improved by cutting the current shortest path through the target.

Take, for example, the graph in Figure~\ref{fig:badCut}. The shortest path from $s$ to $t$ via $e^*$ is 4 hops and uses the sole low-cost edge. If we solve Force Path Cut with this path as $p^*$, all of the $k$ paths along the top of the figure have to be cut. If we remove the low-cost edge, all these paths would be cut and the shortest path from $s$ to $t$ will go through $e^*$. Since that edge is being preserved as part of $p^*$, however, each of the $k$ parallel 2-hop paths has to be cut individually and incur a cost of $\cmax$. For large $k$, we see that this cost is within a constant factor of removing all edges in the graph, when a low-cost solution was available. The existence of this scenario proves the following theorem.
\begin{theorem}
Solving Force Path Cut targeting the shortest path through $e^*$ (or $v^*$) may yield a cost $\Omega(M\cmax/\cmin)$ times greater than solving Force Edge Cut (or Force Node Cut), where $\cmax$ and $\cmin$ are, respectively, the maximum and minimum edge removal costs.
\end{theorem}
\begin{figure}
    \centering
    \includegraphics[width=0.5\textwidth]{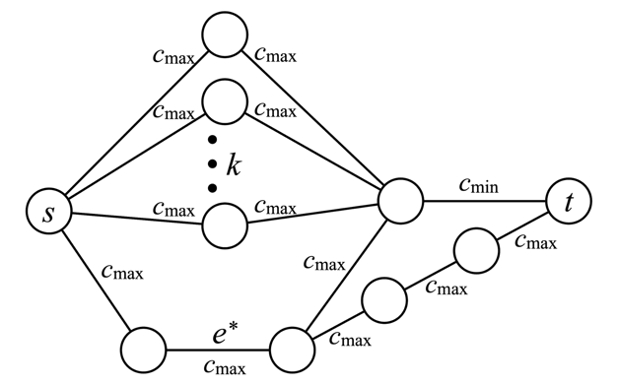}
    \caption{A scenario where preserving the initial shortest path through $e^*$ results in a bad solution. Here, all edges have the same weight and the labels are removal costs. The shortest path through $e^*$ also uses the edge with cost $c_{\min}$. However, if we solve Force Path Cut with this path as $p^*$, the $k$ parallel paths on the top of the figure need to be cut, resulting in a cost of $kc_{\max}$. If, on the other hand, the edge with cost $c_{\min}$ were removed, then the shortest path from $s$ to $t$ would go through $e^*$, at a factor of $k\frac{c_{\max}}{c_{\min}}$ lower cost. If $k$ is $\Omega(M)$, then solution is within a constant factor of removing all edges, when a solution with the lowest possible cost was available.}
    \label{fig:badCut}
\end{figure}

To both minimize cost and allow the flexibility to alter the target path, we formulate a nonconvex optimization problem. We start with a formulation of a linear program to obtain the shortest path through $e^*$ (or $v^*$). 
Recall the linear program formulation of the shortest path problem (see, e.g., \cite{Aneja1978}). Using a graph's incidence matrix, finding the shortest path can be formulated as a linear program. In the incidence matrix, denoted by $\CMat$, each row corresponds to a node and each column corresponds to an edge. The column representing the edge from node $u$ to node $v$ contains $-1$ in the row for $u$, $1$ in the row for $v$, and zeros elsewhere. If the graph is undirected, the edge orientation is arbitrary. To identify the shortest path from $s$ to $t$, define the vector $\dVec\in\{-1,0,1\}^N$, which is $-1$ in the row corresponding to $s$, $1$ in the row for $t$, and zeros elsewhere. The shortest path is the solution to the integer linear program
\begin{align}
    \hat{\xVec}=&\argmin_{\xVec}\xVec^\top\wVec\label{eq:pathLengthObj}\\
    \text{s.t.} &~\xVec\in\{0, 1\}^{M}\label{eq:binaryCut}\\
    &~\CMat\xVec=\dVec\label{eq:pathConstraint}.
\end{align}
The resulting vector $\hat{\xVec}$ is an indicator for the edges in the shortest path.\footnote{In the undirected case, to denote traversal of an edge in the opposite direction of its arbitrary orientation, we consider $\xVec_\textrm{pos}$ and $\xVec_\textrm{neg}$, with constraint~(\ref{eq:pathConstraint}) replaced with $\CMat(\xVec_\textrm{pos}-\xVec_\textrm{neg})=\dVec$ and the objective replaced with $(\xVec_\textrm{pos}+\xVec_\textrm{neg})^\top\wVec$. To restrict traversing an edge in both directions, we add the constraint $\xVec_\textrm{pos}+\xVec_\textrm{neg}\leq\ones$} Constraint (\ref{eq:pathConstraint}) ensures that $\hat{\xVec}$ is a path, and the objective guarantees the result has minimum weight. Due to the structure of the incidence matrix, we can relax the constraint~(\ref{eq:binaryCut}) into the continuous interval $\xVec\in[0,1]^M$. If the shortest path from $s$ to $t$ is unique, $\hat{\xVec}$ will be an indicator vector for the shortest path despite this relaxation. If there are multiple shortest paths, $\hat{\xVec}$ will be a convex combination of the vectors for these paths.

To obtain the shortest path through $e^*$, we perform a similar optimization over two paths: the path from $s$ to $e^*$ and the path from $e^*$ to $t$.\footnote{For brevity, we focus on the solution to Force Edge Cut. The solution to Force Node Cut is similar.} Let $e^*=(u,v)$ be the target edge. (In the undirected setting, we consider both $(u, v)$ and $(v, u)$ independently and choose the lower-cost solution.) Since we consider two paths, we have two vectors $\dVec_1$ and $\dVec_2$, analogous to $\dVec$ in (\ref{eq:pathConstraint}). Assuming that $s\neq u$ and $t\neq v$, $\dVec_1$ is $-1$ in the row of $s$ and $1$ in the row of $u$, with zeros elsewhere, and $\dVec_2$ similarly has $-1$ in the row for $v$ and $1$ in the row for $t$. (If $s=u$, $\dVec_1$ is all zeros, as is $\dVec_2$ if $t=v$.) Since we are optimizing two paths, we use two path vectors, $\xVec_1\in[0,1]^M$ and $\xVec_2\in[0,1]^M$.

It is not, however, sufficient to obtain a path from $s$ to $u$ and one from $v$ to $t$: The concatenation must also be a simple (acyclic) path. To ensure the resulting path has no cycles, we add a constraint to ensure any node is visited at most once. Since the path vectors are optimized over edges, we ensure any node occurs at most twice in the set of edges, aside from the terminal nodes, which occur at most once. To obtain this constraint, we use the matrix $\Cabs$, which contains the absolute values of the incidence matrix, i.e., the $i$th row and $j$th column of $\Cabs$ contains $|\CMat_{ij}|$. We also define $\dabs\in\{1,2\}^N$ to be $1$ in the rows associated with $s$, $t$, $u$, and $v$, and $2$ elsewhere. This yields the linear program 
\begin{align}
    \hat{\xVec}_1, \hat{\xVec}_2=&\argmin_{\xVec_1, \xVec_2}(\xVec_1+\xVec_2)^\top\wVec\\
    \text{s.t.} &~\xVec_1, \xVec_2\in[0, 1]^{M}\\
    &~\CMat\xVec_1=\dVec_1\\
    &~\CMat\xVec_2=\dVec_2\\
    &~\Cabs(\xVec_1+\xVec_2)\leq\dabs.
\end{align}

To solve Force Edge Cut, we use the same constraint generation technique as in \PATHATTACK{}, plus a nonlinear constraint to ensure the returned path does not contain any removed edges. As in \PATHATTACK{}, we use a vector $\Delta\in[0, 1]$. In addition to constraining $\Delta$ to not cut $e^*$, we add the nonconvex bilinear constraint that $\Delta$ is $0$ anywhere $\xVec_1$ or $\xVec_2$ is nonzero. Finally, we again consider a subset of paths $P$ we want to ensure are cut. The resulting nonconvex program
\begin{align}
    \hat{\xVec}_1, \hat{\xVec}_2, \hat{\Delta}=&\argmin_{\xVec_1, \xVec_2, \Delta}(\xVec_1+\xVec_2)^\top\wVec\label{eq:ForceEdgeObj}\\
    \text{s.t.} &~\xVec_1, \xVec_2, \Delta\in[0, 1]^{M}\\
    &~\CMat\xVec_i=\dVec_i,~i\in\{1,2\}\\
    &~\Cabs(\xVec_1+\xVec_2)\leq\dabs\\
    &~\xVec_i^\top\Delta=0,~i\in\{1,2\}\\
    &~\Delta_{e^*}=0\\
    &~\Delta^\top\cVec\leq b\\
    &~\Delta^\top\xVec_p\geq 1~\forall p\in P
\end{align}
provides a partial solution.

The constraint generation mechanism differs slightly from the Force Path Cut case. When solving Force Path Cut, there is a specific target path whose length does not change. For Force Edge Cut, when a new path is added to the constraint set $P$, it may change the shortest uncut path through $e^*$, and thus change the length threshold for inclusion. Thus, we want $P$ to include all paths that not longer than the shortest path in the solution, which is not available until the problem has been solved. As in \PATHATTACK{}, each time we solve the optimization problem, we find the shortest path after removing the edges indicated by $\hat{\Delta}$ as well as $e^*$. If this path is not longer than the shortest uncut path through $e^*$---the path indicated by $\hat{\xVec}_1$, $\hat{\xVec}_2$, and $e^*$---then this alternative path must also be cut, and we add a constraint to achieve this. Algorithm~\ref{alg:e*constraintGen} provides psuedocode for this modified constraint generation procedure. 
\begin{algorithm}
\KwIn{Graph $G=(V, E)$, weights $w$, costs $c$, source $s$, destination $t$, target edge $e^*$, budget $b$}
\KwOut{Set $E^\prime$ of edges to cut}
$\CMat\gets$ unweighted incidence matrix of $G$\;
$\wVec\gets$ weight vector from $w$; $\cVec\gets$ cost vector from $c$\;
$\dVec_1\gets$ $s$ to $e^*$ vector; $\dVec_2\gets$ $e^*$ to $t$ vector\;
$\Cabs\gets |\CMat|$; $\dabs\gets$ ``no cycle'' vector\;
$P\gets\emptyset$\;
 $\mathrm{done}\gets \mathbf{False}$\;
\While{$\mathbf{not}$ $\mathrm{done}$}{
    $\hat{\xVec}_1, \hat{\xVec}_2, \hat{\Delta}\gets$ solution to (\ref{eq:ForceEdgeObj})\;
    $\hat{\Delta}\gets$ extract single path indicator from $\hat{\Delta}$ if not binary\;
    $p_1\gets$ path from $\hat{\xVec}_1$, $e^*$, and $\hat{\xVec}_2$\;
    $E^\prime\gets$ edges with nonzeros in $\hat{\Delta}$\;
    $G^\prime\gets(V, E\setminus (E^\prime\cup\{e^*\}))$\;
    $p_2\gets$ shortest path from $s$ to $t$ in $G^\prime$ (using weights $w$)\;
    \eIf{$p_2$ is not longer than $p_1$}{
        $P\gets P\cup\{p_2\}$\;
    }{
        $\mathrm{done}\gets \mathbf{True}$\;
    }
}
\Return $E^\prime$\;
\caption{Constraint generation procedure for Force Edge Cut.}
\label{alg:e*constraintGen}
\end{algorithm}

To minimize the cost of the attack, we perform a binary search with respect to the budget $b$. We obtain upper and lower bounds for the budget by running \PATHATTACK{} targeting the shortest path through $e^*$. We run the standard \PATHATTACK{} to get the upper bound, and remove the constraint that the target path is uncut for the lower bound, instead only constraining that $e^*$ is not cut. If, during the search, a new upper bound is discovered (i.e., a path through $e^*$ is discovered that requires a budget smaller than the one under consideration), we create new upper and lower bounds based on the new satisfactory path. Algorithm~\ref{alg:binarySearch} outlines this procedure.
\begin{algorithm}
\KwIn{Graph $G=(V, E)$, weights $w$, costs $c$, source $s$, destination $t$, target edge $e^*$, tol. $\epsilon$}
\KwOut{Set $E^\prime$ of edges to cut}
$p\gets$ shortest path from $s$ to $t$ via $e^*$\;
$E_p\gets$ edges in $p$\;
$\Eupper\gets$\PATHATTACK{}$(G, c, w, p, E_p, \texttt{RandPathCover})$\;
$\bupper\gets\sum_{e\in \Eupper}{c(e)}$\;
$\Elower\gets$\PATHATTACK{}$(G, c, w, p, \{e^*\}, \texttt{RandPathCover})$\;
$\blower\gets\sum_{e\in \Elower}{c(e)}$\;
\While{$\bupper-\blower > \epsilon$}{
    $\bmid \gets (\bupper+\blower)/2$\;
    $\Emid\gets$Algorithm~\ref{alg:e*constraintGen}$(G, w, c, s, t, \bmid)$\;
    \eIf{Algorithm~\ref{alg:e*constraintGen} could not be solved}{
        $\blower\gets\bmid$;\ \ \ \ \ \tcp{budget is too small}
        remove constraints accrued since last iteration
    }{
        \tcp{budget is sufficient}
        $G^\prime\gets(V, E\setminus \Emid)$\;
        $p\gets$ shortest path from $s$ to $t$ in $G^\prime$;~~\tcp{this path will include $e^*$}
        $\Etemp\gets$\PATHATTACK{}$(G, c, w, p, E_p, \texttt{RandPathCover})$\;
        $\btemp\gets\sum_{e\in \Etemp}{c(e)}$\;
        \eIf{$\btemp < \bmid$}{
            $\bupper\gets\btemp$\;
            $\Eupper\gets\Etemp$\;
        }{
            $\bupper\gets\bmid$\;
            $\Eupper\gets\Emid$\;
        }
        $\Etemp\gets$\PATHATTACK{}$(G, c, w, p, \{e^*\}, \texttt{RandPathCover})$\;
        $\btemp\gets\sum_{e\in \Etemp}{c(e)}$\;
        \If{$\btemp > \blower$}{
            $\blower\gets\btemp$\;
            $\Elower\gets\Etemp$\;
        }
        
    }
}
\Return $\Eupper$;~~\tcp{return the best valid solution found}
\caption{Combinatorial Search}
\label{alg:binarySearch}
\end{algorithm}

While Algorithm~\ref{alg:binarySearch} will converge to within $\epsilon$ of an optimal solution in a logarithmic number of iterations, each iteration uses Algorithm~\ref{alg:e*constraintGen}, which involves solving a nonconvex optimization problem. It could, therefore, require exponential time to complete any given iteration, if finding an optimal solution is required. Thus, in addition to the combinatorial optimization method, we consider a heuristic algorithm that seeks to identify the bottlenecks that prevent an optimal solution, as illustrated in Figure~\ref{fig:badCut}, and that is guaranteed to run in polynomial time.

As in the combinatorial optimization, we leverage \PATHATTACK{} with and without constraints to avoid cutting the target path. In this case, after running \PATHATTACK{} while only preventing $e^*$ from being cut, we consider the edges on the target path that are cut by \PATHATTACK{}. For each of these edges, we consider the possibility of either (1) removing the edge from the graph entirely, or (2) marking it to never be removed. If removal of any of these edges causes $t$ to become unreachable from $s$, we add that edge to a list of uncuttable edges. Otherwise, we consider the case in which each of these edges is removed, each time finding the shortest path through $e^*$. We run \PATHATTACK{} in each case, and find the edge whose removal results in the lowest upper bound on the removal budget. This edge is added to a list of edges that will always be removed. This procedure terminates when the upper and lower bounds converge. A pseudocode description of this heuristic search is provided in Algorithm~\ref{alg:heuristicSearch}.

\begin{algorithm}
\KwIn{Graph $G=(V, E)$, weights $w$, costs $c$, source $s$, destination $t$, target edge $e^*$, tol. $\epsilon$}
\KwOut{Set $E^\prime$ of edges to cut}
$\Ealways\gets\emptyset$\;
$\Enever\gets\emptyset$\;
$\Ebest\gets\emptyset$\;
$\cbest\gets\infty$\;
\Repeat{$\cbest \leq \blower+\epsilon$}{
    $G^\prime\gets(V, E\setminus \Ealways)$\;
    $p\gets$ shortest path from $s$ to $t$ via $e^*$ in $G^\prime$\;
    $E_p\gets$ edges in $p$\;
    $\Eupper\gets$\PATHATTACK{}$(G, c, w, p, E_p, \texttt{RandPathCover})$\;
    $\bupper\gets\sum_{e\in \Eupper}{c(e)}$\;
    $\Elower\gets$\PATHATTACK{}$(G^\prime, c, w, p, \{e^*\}\cup\Enever, \texttt{RandPathCover})$\;
    $\blower\gets\sum_{e\in \Elower\cup\Ealways}{c(e)}$\;
    \If{$\bupper< \cbest$}{
        $\cbest\gets\bupper$\;
        $\Ebest\gets\Eupper$\;
    }
    \If{$\blower<\bupper$}{
        budget $\gets$ empty hash table\;
        \For{$e\in E_p\cap\Elower$}{
            remove $e$ from $G$\;
            $E_1\gets$\PATHATTACK{}$(G, c, w, p, E_p, \texttt{RandPathCover})$\;
            budget[$e$]$ \gets c(e)+\sum_{e_1\in E_1}{c(e_1)}$\;
            add $e$ to $G$\;
        }
        \eIf{$\exists e\in E_p\cap\Elower$ where removing $e$ disconnects $s$ and $t$}{
        $\Enever\gets\Enever\cup{\{e\}}$\;
        }{
        $e_{\mathrm{new}}\gets\argmin_{e\in E_p\cap\Elower}{\mathrm{budget}[e]}$\;
        $\Ealways\gets\Ealways\cup\{e\}$\;
        }
    }
}
\Return $\Ebest$\;
\caption{Heuristic Search}
\label{alg:heuristicSearch}
\end{algorithm}

By Theorem~\ref{thm:PATHATTACKconvergence}, PATHATTACK converges in polynomial time. Each iteration of the algorithm considers at most $N-1$ edges (the maximum number of edges in a path) and runs \PATHATTACK{} twice for each. At each iteration, at least one edge is added to either (1) the set of edges to never cut or (2) the set of edges to always cut. The total number of iterations is, therefore, bounded by the number of edges. This means that \PATHATTACK{} will be run $O(MN)$ times over the course of Algorithm~\ref{alg:heuristicSearch}, which yields the following theorem.
\begin{theorem}
Algorithm~\ref{alg:heuristicSearch} completes in polynomial time.
\end{theorem}

\section{Experiments}
\label{sec:experiments}

This section presents baselines, datasets,  experimental setup, and results. 

\subsection{Baseline Methods}
We consider two simple greedy methods as baselines for assessing performance of \PATHATTACK{}. Each of these algorithms iteratively computes the shortest path $p$ between $s$ and $t$; if $p$ is not longer than $p^*$, it uses some criterion to cut an edge from $p$. When we cut the edge with minimum cost, we refer to the algorithm as \texttt{GreedyCost}. We also consider a version where we cut the edge in $p$ with the largest ratio of eigenscore\footnote{The eigenscore of an edge is the product of the entries in the principal eigenvector of the adjacency matrix corresponding to the edge's vertices.} to cost, since edges with high eigenscores are known to be important in network  flow~\cite{tongCIKM2012}. This version of the algorithm is called \texttt{GreedyEigenscore}. In both cases, edges from $p^*$ are not allowed to be cut. The baseline method to solve Force Path Cut (or Force Node Cut) is to identify the shortest path through $e^*$ (or $v^*$), use this path as $p^*$, and solve \PATHATTACK{}.

\subsection{Synthetic and Real Networks}
\label{sec:graphdata}
We use both synthetic and real networks in our  experiments. For the synthetic networks, we run seven different graph models to generate 100 synthetic networks of each model. We pick parameters to yield networks with similar numbers of edges ($\approx 160$K). We use 16,000-node Erd\H{o}s--R\'{e}nyi graphs, both undirected (ER) and directed (DER), with edge probability $0.00125$, 16,000-node Barab\'{a}si--Albert (BA) graphs with average degree 20, 16,000-node Watts--Strogatz (WS) graphs with average degree 20 and rewiring probability $0.02$, $2^{14}$-node stochastic Kronecker graphs (KRON), $285\times285$ lattices (LAT), and 565-node complete graphs (COMP).

We use seven weighted and unweighted networks. The unweighted networks are the Wikispeedia graph (WIKI)~\cite{West2009}, an Oregon autonomous system network (AS)~\cite{Leskovec2005}, and a Pennsylvania road network (PA-ROAD)~\cite{Leskovec2009}. The weighted networks are Central Chilean Power Grid (GRID)~\cite{Kim2018}, Lawrence Berkeley National Laboratory network data (LBL), the Northeast US Road Network (NEUS), and the DBLP coauthorship graph (DBLP)~\cite{Benson2018}.  All real networks are undirected except for WIKI and LBL. The networks range from 444 edges on 347 nodes to over 8.3M edges on over 1.8M nodes, with average degree ranging from 2.5 to  46.5 and triangle count ranging from 40 to nearly 27M. Further details on the real and synthetic networks---including URLs to the real data---are provided in Appendix~\ref{sec:datasets}.

For the synthetic networks and unweighted real networks, we consider three different edge-weight assignment schemes: Poisson, uniform random, or equal weights. For Poisson weights, each edge $e$ has an independently random weight $w_e=1+w_e^\prime$, where $w_e^\prime$ is drawn from a Poisson distribution with rate parameter 20. For uniform weights, each weight is drawn from a discrete uniform distribution of integers from 1 to 41. This yields the same average weight as Poisson weights.

\subsection{Experimental Setup}
\label{subsec:setup}
For each graph---considering graphs with different edge-weighting schemes as distinct---we run 100 experiments. In each experiment, we select $s$ and $t$ uniformly at random among all nodes, with the exception of LAT, PA-ROAD, and NEUS, where we select $s$ uniformly at random and select $t$ at random among nodes 50 hops away from $s$.\footnote{This alternative method of selecting the destination was used due to the computational expense of identifying successive shortest paths in large grid-like networks.} Given $s$ and $t$, we identify the shortest simple paths and use the 100th, 200th, 400th, and 800th shortest as $p^*$ in four experiments. For the large grid-like networks (LAT, PA-ROAD, and NEUS), this procedure is run using only the 60-hop neighborhood of $s$. We focus on the case where the edge removal cost is equal to the weight (distance). For Force Edge Cut and Force Node Cut, we consider consecutive shortest paths from $s$ to $t$ until we see 5 edges (or nodes) not on the initial shortest path. The 5th edge (or node) we see that was not on the original shortest path is used as $e^*$ (or $v^*$).

When running Algorithm~\ref{alg:e*constraintGen}, we let the nonconvex optimization (\ref{eq:ForceEdgeObj}) run no more than 10 minutes. If a feasible point is not found in that time, we assume the model is infeasible and increase the budget. We stop the procedure every 30 seconds to check the best candidate solution and see if it satisfies our criteria. In addition, we stop the optimization if the objective matches the length of our best incumbent solution. If we consider a budget for over 8 hours, we consider it infeasible and increase the lower bound. The entire procedure is terminated, and the lowest upper bound returned, if it is still running after 24 hours.

The experiments were run on Linux machines with 32 cores and 192 GB of memory. The LP in \texttt{PATHATTACK-Rand} was implemented using Gurobi 9.1.1, and shortest paths were computed using  \texttt{shortest\_simple\_paths} in NetworkX.\footnote{Gurobi is at \url{https://www.gurobi.com}. NetworkX is at \url{https://networkx.org}. Code from the experiments is at \url{https://github.com/bamille1/PATHATTACK}.} The combinatorial search method is set to use 4 threads.

\subsection{PATHATTACK Results}
\label{sec:results}

We treat the result of \texttt{GreedyCost} as our baseline cost and report the cost of other algorithms' solutions as a reduction from the baseline. With one exception\footnote{\texttt{GreedyEigenscore} only outperforms \texttt{GreedyCost} in COMP with uniform weights.}, \texttt{GreedyCost} outperforms \texttt{GreedyEigen\-score} in both running time and edge removal cost, so we omit the \texttt{GreedyEigenscore} results for clarity of presentation. In all experiments, we edge removal costs equal to the weights.  Fig.~\ref{fig:pathUnweighted} shows the results on both synthetic and real unweighted graphs, which have had synthetic weights added to the edges. Fig.~\ref{fig:pathWeighted} shows the results on real weighted networks. In these figures, the 800th shortest path is used as $p^*$; other results were similar and omitted for brevity.

\begin{figure*}
    \centering
    \includegraphics[width=\textwidth]{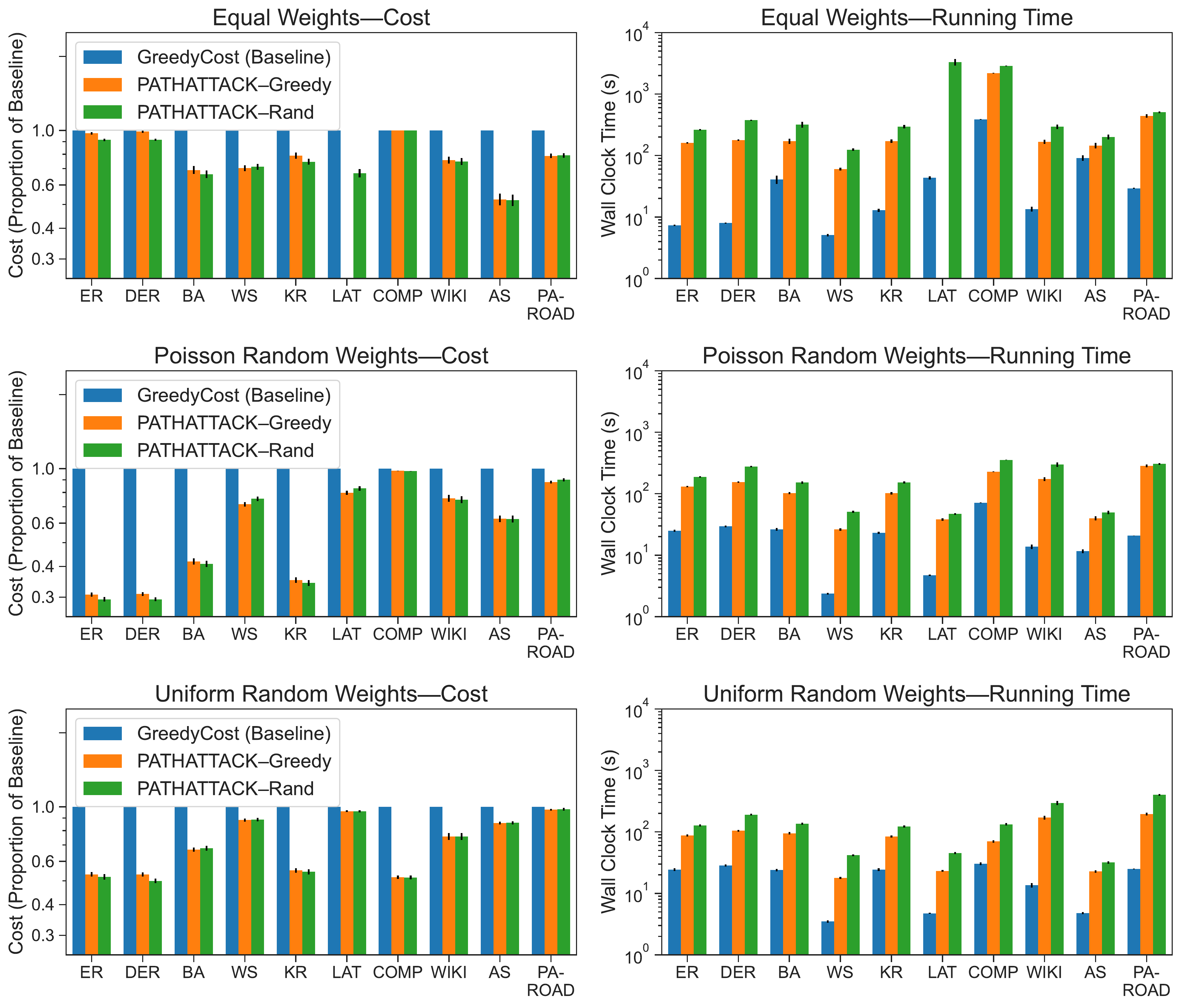}
    \caption{\PATHATTACK{} results on synthetic and real unweighted graphs with synthetic weights added, in terms of edge removal cost (left column) and running time (right column). Lower is better for both metrics. Cost is plotted as a proportion of the cost using the \texttt{GreedyCost} baseline. Bar heights are means across 100 trials and error bars are standard errors. Results are shown with equal weights on all edges (top row) and edge weights drawn from Poisson (middle row) and Uniform (bottom row) distributions. \texttt{PATHATTACK-Greedy} operated on lattices with equal weights for over one day without converging, so results were not collected for this case. \PATHATTACK{} yields a substantial reduction in cost in ER, BA, KR, WIKI, and AS graphs, while the baseline is often near-optimal for LAT and PA-ROAD.}
    \label{fig:pathUnweighted}
\end{figure*}

\begin{figure*}
    \centering
    \includegraphics[width=\textwidth]{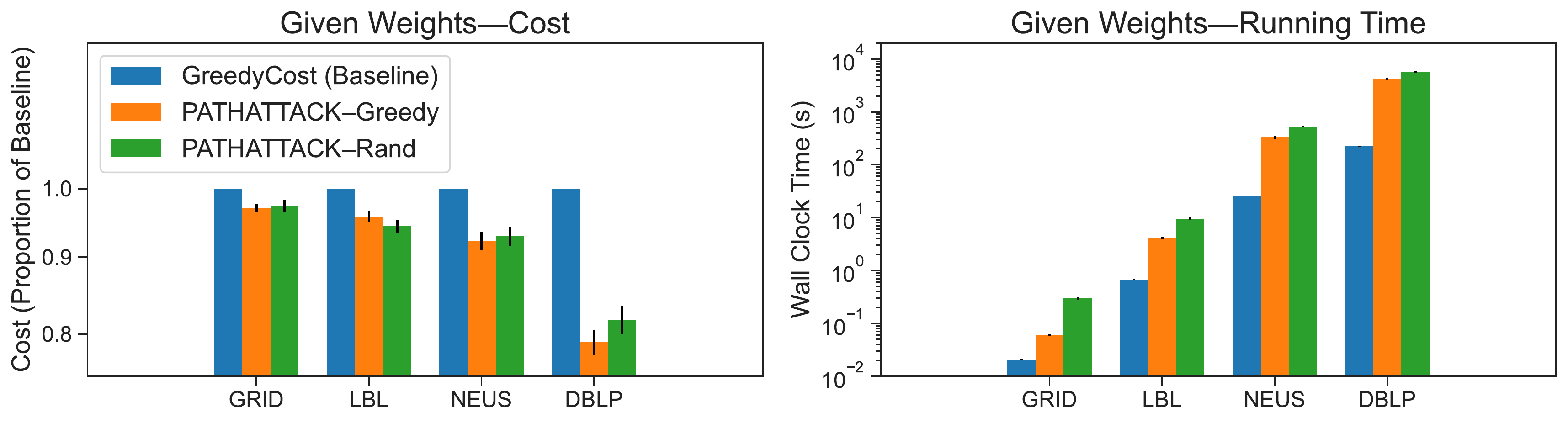}
    \caption{\PATHATTACK{} results on real weighted graphs in terms of edge removal cost (left) and running time (right). Lower is better for both metrics. Cost is plotted as a proportion of the cost using the \texttt{GreedyCost} baseline. Bar heights are means across 100 trials and error bars are standard errors. Note the difference in scale on the vertical axes from Fig.~\ref{fig:pathUnweighted}. \PATHATTACK{} yields a substantial improvement in performance for DBLP, while the baseline performs well on the other (less clustered) networks.}
    \label{fig:pathWeighted}
\end{figure*}

Comparing the cost achieved by \texttt{PATHATTACK} to those obtained by the greedy baseline, we observe some interesting phenomena. Lattices and road networks, for example, have a similar tradeoff: \texttt{PATHATTACK} provides a mild improvement in cost at the expense of an order of magnitude additional processing time. Considering that \texttt{PATHATTACK-Rand} usually results in the optimal solution (over 86\% of the time), this means that the baselines often achieve near-optimal cost with a na\"{i}ve algorithm. On the other hand, ER, BA, and KR graphs follow a trend more similar to the AS and WIKI networks, particularly in the randomly weighted cases: The cost is cut by a substantial fraction---enabling the attack with a smaller budget---for a similar or smaller time increase. This suggests that the time/cost tradeoff is much less favorable for less clustered, grid-like networks (note the clustering coefficients in Appendix~\ref{sec:datasets}). 

Cliques (COMP) are particularly interesting in this case, showing a phase transition as the entropy of the weights increases. When edge weights are equal, cliques behave like an extreme version of the road networks: an order of magnitude increase in run time with no decrease in cost. With Poisson weights, \texttt{PATHATTACK} yields a slight improvement in cost, whereas when uniform random weights are used, the clique behaves much more like an ER or BA graph. In the unweighted case, $p^*$ is a three-hop path, so all other two- and three-hop paths from $s$ to $t$ must be cut, which the baseline does efficiently. Adding Poisson weights creates some randomness, but most edges have a weight that is about average, so it is still similar to the unweighted scenario. With uniform random weights, we get the potential for much different behavior (e.g., short paths with many edges) for which the greedy baseline's performance suffers.

There is an opposite, but milder, phenomenon with PA-ROAD and LAT: using higher-entropy weights \emph{narrows} the cost difference between the baseline and \texttt{PATHATTACK}. This may be due to the source and destination being many hops away. With the terminal nodes many hops apart, many shortest paths between them could go through a few low-weight (thus low-cost) edges. A very low weight edge between two nodes would be very likely to occur on many of the shortest paths, and would be found in an early iteration of the greedy algorithm and removed, while considering more shortest paths at once would yield a similar result. 
We also note that, in the weighted graph data, LBL and GRID behave similarly to road networks. Among our real datasets, these have a low clustering coefficient 
(see Appendix~\ref{sec:datasets}). This lack of overlap in nodes' neighborhoods may lead to better relative performance with the baseline, since there may not be a great deal of overlap between candidate paths. 

\subsection{Results Targeting a Node}
In this section, we present results for the case where the adversary targets a node $v^*$. The baseline in these experiments is to run \PATHATTACK{} using the shortest path from $s$ to $t$ through $v^*$ as $p^*$. We call this method \PATHATTACK{}-$v^*$. We present results using this method along with results for the heuristic search method (Algorithm~\ref{alg:heuristicSearch}) and combinatorial optimization.

In most cases, all three methods yield a solution of the same cost. To clarify the performance differences, we separate these cases from those where the costs differ between methods. For the case where all methods result in the same cost, we show running time results on unweighted graphs in Fig.~\ref{fig:nodeUnweighted_same}. As when targeting paths, lattices and road networks are similar: they are the only graphs where the heuristic methods do not match the result of the combinatorial optimization a majority of the time, largely due to the equal-weight case. Watts--Strogtaz graphs, which have a lattice-like component, also frequently have different results across methods. Fig.~\ref{fig:nodeUnweighted_different} illustrates cases where not all algorithms yield the same cost. Again, lattices and roads are distinct: here they see a much more substantial improvement from the combinatorial optimization than heuristics, with Watts--Strogatz graphs also sometimes being similar. There are many cases where the heuristic search yields the same result as the \PATHATTACK{}-$v^*$ baseline. Upon inspection, many of these cases result from multiple shortest paths using $v^*$: The heuristic overlooks one solution because it does not cut the current $p^*$, which results in no edges to consider in the inner loop of Algorithm~\ref{alg:heuristicSearch}.
\begin{figure*}[t]
    \centering
    \includegraphics[width=0.6\textwidth]{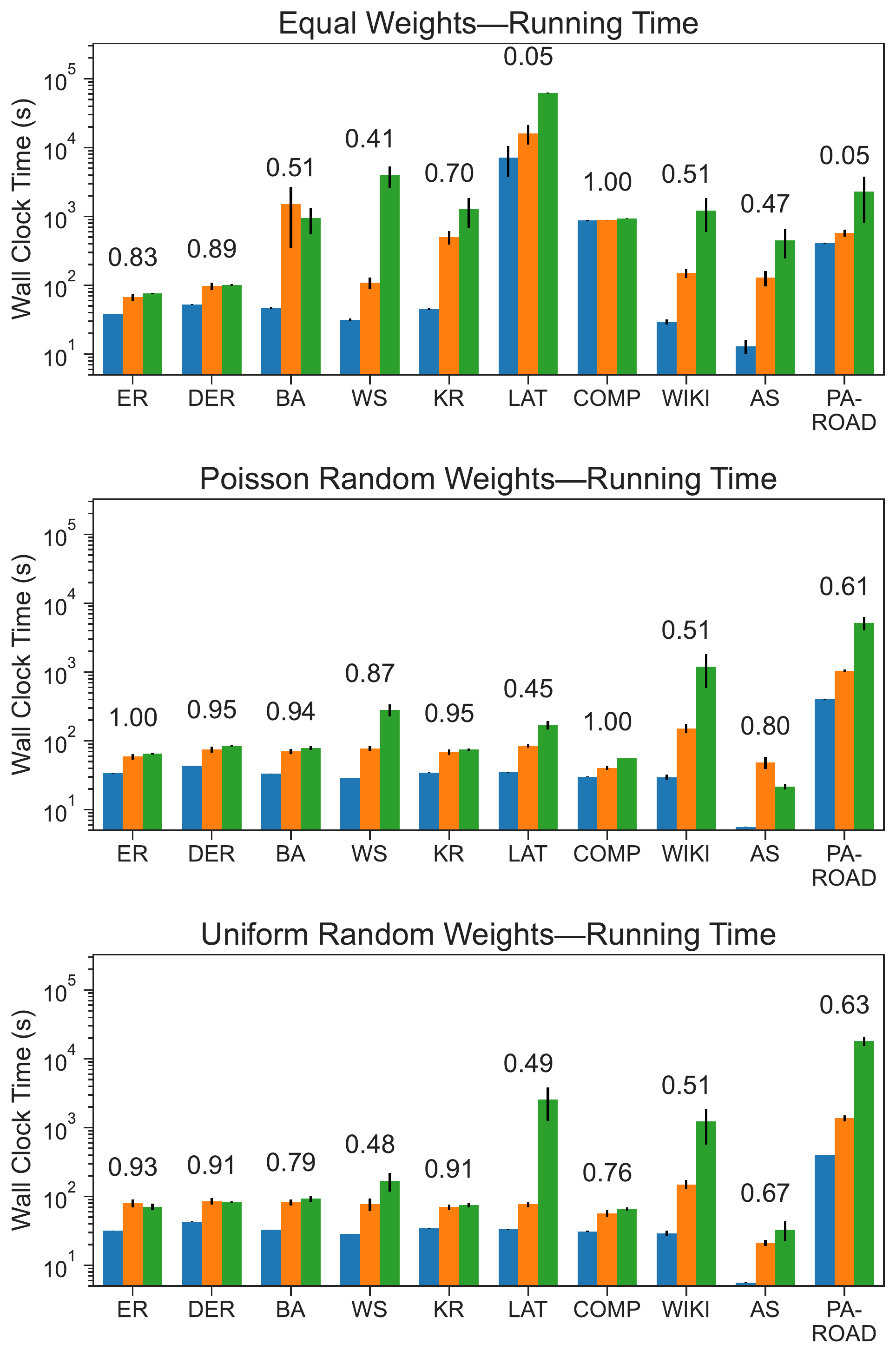}
    \caption{Running time results on unweighted networks when targeting a specific node in cases where all three algorithms yield the same cost. Lower time is better. Annotations are the proportion of 100 trials where all costs are the same. Bar heights are means across these trials and error bars are standard errors. Results are shown with equal weights on all edges (top) and edge weights drawn from Poisson (middle) and Uniform (bottom) distributions. With the exception of LAT and PA-ROAD, the baseline method matches the combinatorial optimization in a majority of cases and has a substantially smaller time requirement.}
    \label{fig:nodeUnweighted_same}
\end{figure*}
\begin{figure*}[t]
    \centering
    \includegraphics[width=\textwidth]{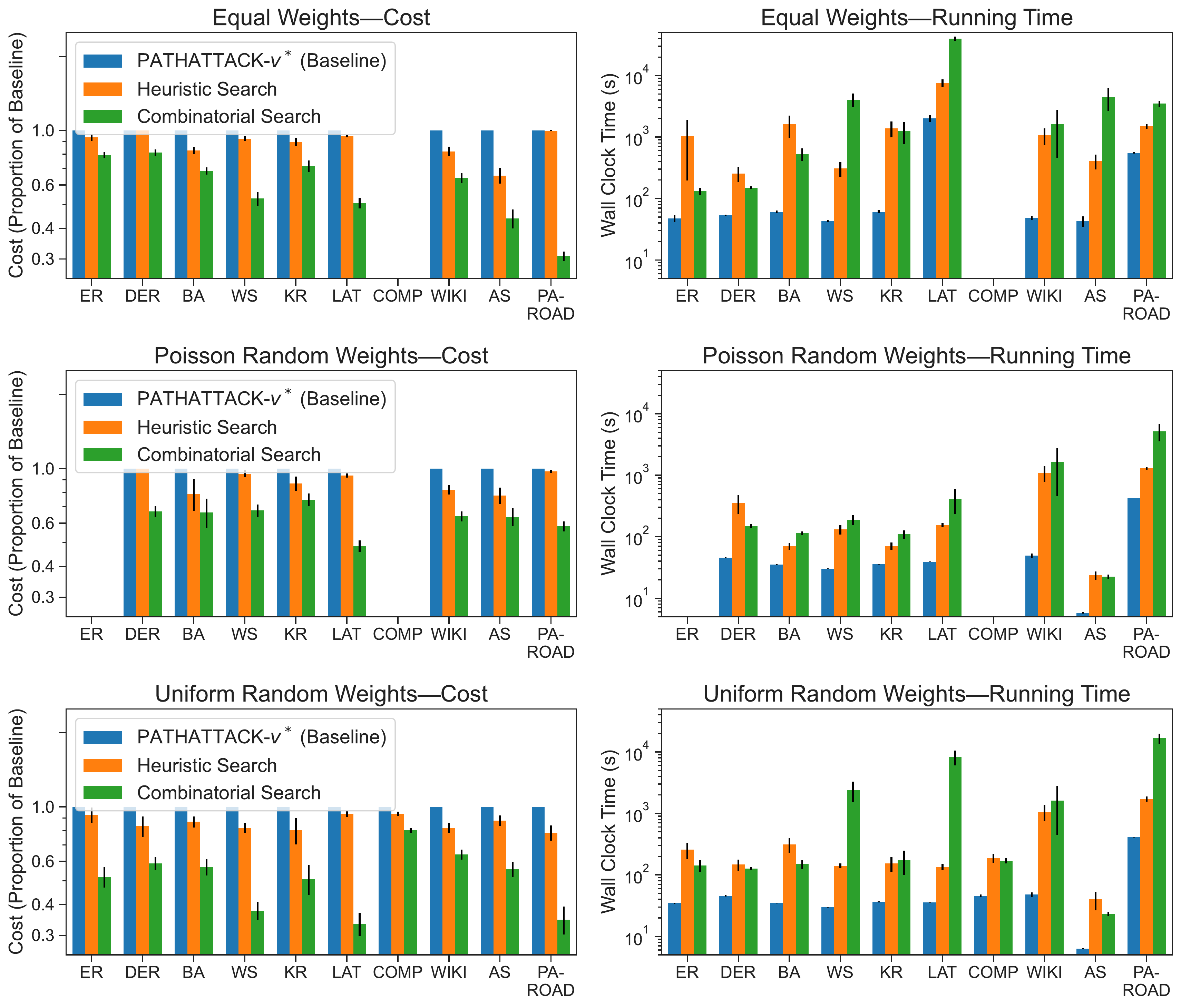}
    \caption{Results on unweighted networks when targeting a specific node in cases where not all node-targeting algorithms yield the same cost, in terms of edge removal cost (left column) and running time (right column). Lower is better for both metrics. Bar heights are means across these trials and error bars are standard errors. Results are shown with equal weights on all edges (top row) and edge weights drawn from Poisson (middle row) and Uniform (bottom row) distributions. Note that COMP with equal and Poisson weights and ER with Poisson weights are not included, as the methods matched in all trials. The greatest cost reductions from using combinatorial search often coincide with large running time increases.  }
    \label{fig:nodeUnweighted_different}
\end{figure*}

For weighted networks, running time in cases where all methods yield the same cost is plotted in Fig.~\ref{fig:nodeWeighted_same}. In all cases, the three algorithms usually find equal-cost solutions, though the combinatorial method frequently times out. The power grid network has a particularly large increase in running time when using the combinatorial method. In cases where not all methods yield the same cost, shown in Fig.~\ref{fig:nodeWeighted_different}, the heuristic search achieves the same cost as the combinatorial search more often than in the unweighted graphs. In fact, on DBLP, the combinatorial search typically times out and the heuristic search outperforms it.
\begin{figure*}[t]
    \centering
    \includegraphics[width=0.6\textwidth]{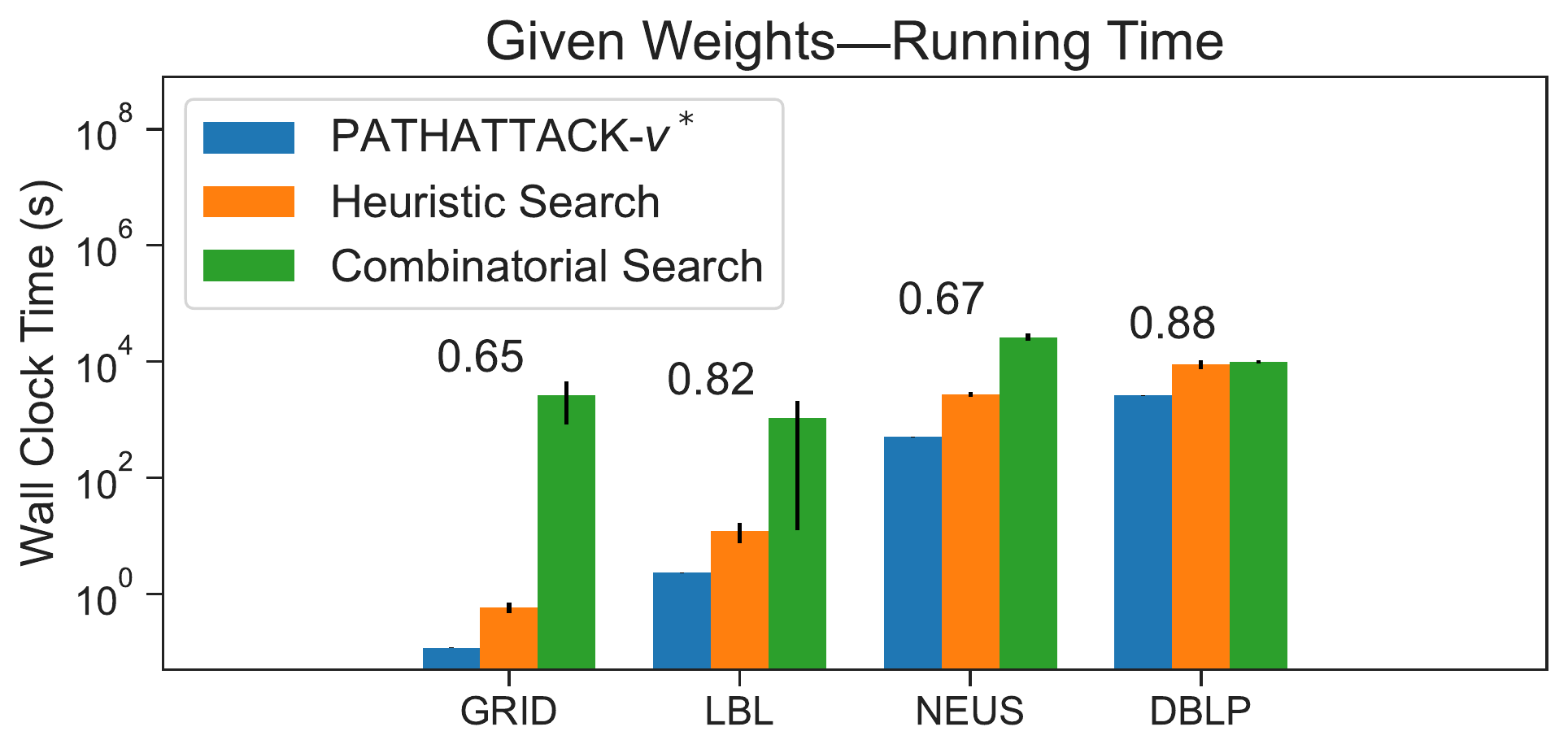}
    \caption{Running time results on weighted networks when targeting a specific node in cases where all three algorithms yield the same cost. Lower time is better. Annotations are the proportion of 100 trials where all costs are the same. Bar heights are means across these trials and error bars are standard errors. In contrast to the unweighted networks, the heuristic methods match the combinatorial optimization in a large majority of cases, even in grid-like networks.}
    \label{fig:nodeWeighted_same}
\end{figure*}
\begin{figure*}[t]
    \centering
    \includegraphics[width=\textwidth]{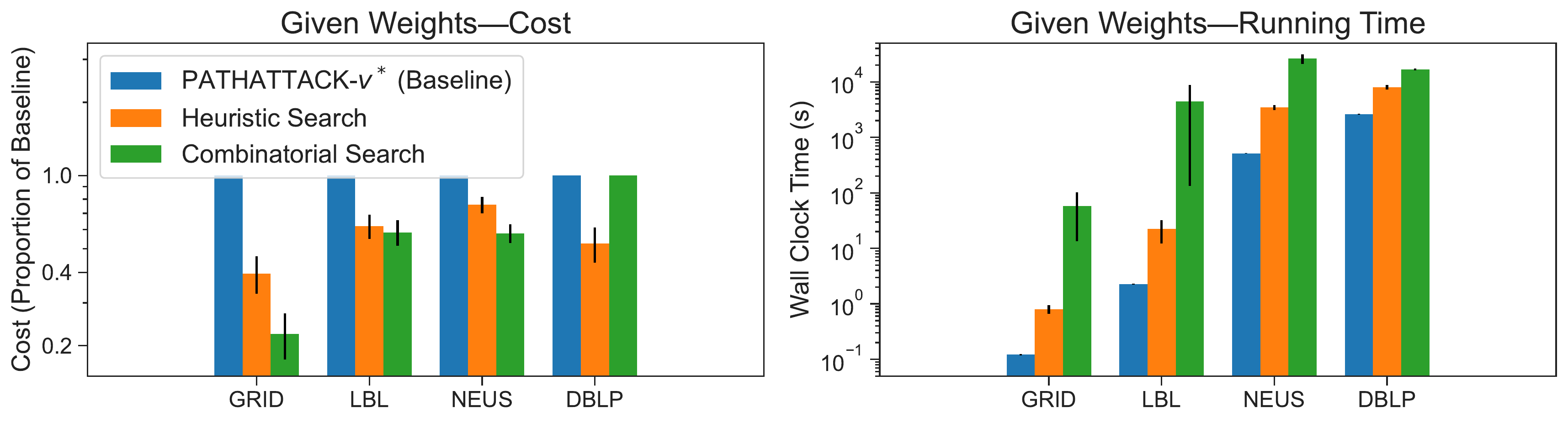}
    \caption{Results on weighted networks when targeting a specific node in cases where not all node-targeting algorithms yield the same cost, in terms of edge removal cost (left) and running time (right). Lower is better for both metrics. Bar heights are means across these trials and error bars are standard errors. Heuristic search in these cases is competitive with combinatorial search, even outperforming it sometimes due to timeouts.}
    \label{fig:nodeWeighted_different}
\end{figure*}

\section{conclusion}
\label{sec:conclusion}
In this paper, we introduce the Force Path Cut, Force Edge Cut, and Force Node Cut problems, in which edges are removed from a graph in order to make the shortest path from one node to another, respectively, be a specific path, use a specific edge, or use a specific node. We show that all three problems are hard to approximate, but that a logarithmic approximation exists for Force Path Cut via existing approximation algorithms for Set Cover. We leverage these methods to develop a new algorithm called \PATHATTACK{} and demonstrate its efficacy in solving Force Path Cut using a thorough set of experiments on real and synthetic networks. We also use \PATHATTACK{} as part of a heuristic search method to solve Force Edge Cut and Force Node Cut, which yields performance similar to a more much more computationally intensive combinatorial search.

There is a gap between the approximation factor of \PATHATTACK{} and the lower bound implied by APX-hardness; it remains an open problem whether Force Path Cut is APX-complete or if there is no constant-factor approximation. The best possible polynomial-time approximation factor for 3-Terminal Cut is $12/11$~\cite{Cunningham1999}, and our results imply that the best such approximation for Force Path Cut is between $12/11$ and  logarithmic. Future work will include defenses against \PATHATTACK{} to make networks more robust against adversarial manipulation.

\begin{acks}
BAM was supported by the United States Air Force under Contract No. FA8702-15-D-0001. TER was supported in part by the Combat Capabilities Development Command Army Research Laboratory (under Cooperative Agreement No.~W911NF-13-2-0045) and by the Under Secretary of Defense for Research and Engineering under Air Force Contract No.~FA8702-15-D-0001. YV was supported by grants from the Army Research Office (W911NF1810208, W911NF1910241) and National Science Foundation (CAREER Award IIS-1905558). Any opinions, findings, conclusions or recommendations expressed in this material are those of the authors and should not be interpreted as representing the official policies, either expressed or implied, of the funding agencies
 or the U.S.~Government. The U.S.~Government is authorized to reproduce and distribute reprints for Government purposes not withstanding any copyright notation here on.
\end{acks}

\bibliographystyle{ACM-Reference-Format}
\bibliography{bibfile}

\appendix

\section{Proof of Theorem~\ref{thm:pathHard}}
\label{sec:proof}
\subsection{Proof for Undirected Graphs}
\label{subsec:proofUndirected}
As noted in Section~\ref{subsec:sketch}, to prove Theorem~\ref{thm:pathHard}, we first prove that Force Path Cut is APX-hard for undirected graphs.
\begin{lemma}
Force Path Cut is APX-hard for undirected graphs, including the case where all weights and all costs are equal.
\label{lem:undirected}
\end{lemma}
To prove Lemma~\ref{lem:undirected}, we reduce 3-Terminal Cut to Force Path Cut via a linear reduction. Let $G=(V,E)$ be an undirected graph, where all weights are equal. We also have three terminal nodes $s_1$, $s_2$, and $s_3\in V$. Since we are proving the problem is APX-hard, we consider the optimization version of 3-Terminal Cut and Force Path Cut, where the goal is to minimize the budget. Thus, the goal of 3-Terminal Cut is to find the smallest set $E^\prime$ such that $s_1$, $s_2$, and $s_3$ are disconnected in $G^\prime=(V, E\setminus E^\prime)$. Dahlhaus et al. show in~\cite{Dahlhaus1994} that 3-Terminal Cut is APX-hard, even when all weights are equal.\footnote{More specifically, it is proved in \cite{Dahlhaus1994} that the problem is MAX SNP-hard, but this implies APX-hardness: if a problem is MAX SNP-hard, it has no polynomial time approximation scheme unless P=NP.}

We propose a linear reduction from 3-Terminal Cut to Force Path Cut. As discussed in~\cite{Dahlhaus1994}, a linear reduction from problem $A$ to problem $B$ consists of two functions $f$ and $g$, where $f$ maps an instance of $A$ to an instance of $B$, and $g$ maps a solution of $B$ to a solution of $A$. To be a linear reduction, the following conditions must hold:
\begin{enumerate}
    \item The functions $f$ and $g$ can be computed in polynomial time.\label{item:polynomial}
    \item Given $\instA$, an instance of problem $A$, the optimal solution to $\instB=f(\instA)$ must be at most  $\alpha$ times the optimal solution for $\instA$, for a constant $\alpha>0$, i.e., $\opt(\instB)\leq\alpha\cdot\opt(\instA)$.\label{item:scaleOpt}
    \item Given a solution $y$ to $\instB=f(\instA)$, $x=g(y)$ is a solution to $\instA$ such that $$|\cost(x)-\opt(\instA)|\leq \beta|\cost(y)-\opt(\instB)|$$ for a constant $\beta>0$.\label{item:scaleSol}
\end{enumerate}

We start by defining $f$, the function from an instance of 3-Terminal Cut to an instance of Force Path Cut. We are given an instance of 3-Terminal Cut as described above. Let $N=|V|$ and $M=|E|$. As shown in Figure~\ref{fig:reduction}, we add $M+1$ new paths of length $N$ from $s_1$ to $s_2$, and the same from $s_2$ to $s_3$. We add a single path of length $2N-1$ from $s_1$ to $s_3$. Algorithm~\ref{alg:problemMap} provides pseudocode for this procedure. 
\begin{algorithm}
\KwIn{Graph $G=(V, E)$, terminals $s_1,s_2,s_3$}
\KwOut{Graph $\hat{G}$, target path $p^*$}
$\hat{V}\gets V$; $N\gets|V|$\;
$\hat{E}\gets E$; $M\gets|E|$\;
\tcp{Create paths from $s_1$ to $s_2$}
\For{$i\gets 1$ to $M+1$}{
    $v_\mathrm{prev}\gets s_1$\;
    \For{$j\gets 1$ to $N-1$}{
        $v_{ij,1}\gets$ new node\;
        $\hat{V}\gets \hat{V}\cup\{v_{ij,1}\}$\;
        $\hat{E}\gets \hat{E}\cup\{\{v_\mathrm{prev}, v_{ij,1}\}\}$\;
        $v_\mathrm{prev}\gets v_{ij,1}$\;
    }
    $\hat{E}\gets \hat{E}\cup\{\{v_\mathrm{prev}, s_2\}\}$\;
}
\tcp{Create paths from $s_2$ to $s_3$}
\For{$i\gets 1$ to $M+1$}{
    $v_\mathrm{prev}\gets s_2$\;
    \For{$j\gets 1$ to $N-1$}{
        $v_{ij,2}\gets$ new node\;
        $\hat{V}\gets \hat{V}\cup\{v_{ij,2}\}$\;
        $\hat{E}\gets \hat{E}\cup\{\{v_\mathrm{prev}, v_{ij,2}\}\}$\;
        $v_\mathrm{prev}\gets v_{ij,2}$\;
    }
    $\hat{E}\gets \hat{E}\cup\{\{v_\mathrm{prev}, s_3\}\}$\;
}
\tcp{Create $p^*$ (a path from $s_1$ to $s_3$)}
$v_\mathrm{prev}\gets s_1$\;
$p^*\gets$ empty path\;
\For{$j\gets 1$ to $2N-2$}{
    $v_{ij,3}\gets$ new node\;
    $\hat{V}\gets \hat{V}\cup\{v_{ij,3}\}$\;
    $\hat{E}\gets \hat{E}\cup\{\{v_\mathrm{prev}, v_{ij,3}\}\}$\;
    append $\{v_\mathrm{prev}, v_{ij,3}\}$ to the end of $p^*$\;
    $v_\mathrm{prev}\gets v_{ij,3}$\;
}
$\hat{E}\gets \hat{E}\cup\{\{v_\mathrm{prev}, s_3\}\}$\;
append $\{v_\mathrm{prev}, s_3\}$ to the end of $p^*$\;

\Return $\hat{G}=(\hat{V}, \hat{E})$, $s_1$, $s_3$, $p^*$\;
\caption{Mapping from 3-Terminal Cut to Force Path Cut.}
\label{alg:problemMap}
\end{algorithm}

Applying Algorithm~\ref{alg:problemMap} to an instance of 3-Terminal Cut $(G, s_1, s_2, s_3)$, we get $(G^\prime, s_1, s_3, p^*)$, an instance of Force Path Cut. Note that, in these instances, all edge weights and removal costs are equal to 1. From the instance of Force Path Cut, we get a solution $E^\prime$. We also have a function that maps a solution to Force Path Cut to a solution to 3-Terminal Cut: include all edges in $E^\prime$ that existed in the original graph, i.e.,
\begin{equation}
    g(E^\prime)=\begin{cases}
    E^\prime\cap E&\textrm{if }|E^\prime| < M\\
    E&\textrm{otherwise}
    \end{cases}
\end{equation}
where $E^\prime$ is the solution to Force Path Cut and $E$ is the original edge set from the 3-Terminal Cut instance. (The edge set $E$ is not a parameter of $g$, as it is fixed within the context of the problem.)

We can see that both functions satisfy condition (\ref{item:polynomial}): The function $g$ simply removes up to $M$ edges from a set, and the body of each loop takes constant time in Algorithm~\ref{alg:problemMap}, and there are two nested loops taking $O(MN)$ time and a final loop taking $O(N)$ time. To show that this reduction satisfies condition (\ref{item:scaleOpt}), we first prove the following Lemma.
\begin{lemma}
Let $\instA$ be an instance of 3-Terminal Cut and $E^\prime$ be a solution to $\instA$. Then $E^\prime$ is also a solution to $f(\instA)$.
\label{lem:3TCSolvesFPC}
\end{lemma}
\begin{proof}
Since $E^\prime$ is a solution to $\instA$, the graph $G^\prime=(V, E\setminus E^\prime)$ has at least three connected components, where one contains $s_1$, one contains $s_2$, and one contains $s_3$. The edges $\Enew$ added by Algorithm~\ref{alg:problemMap} create paths between the connected components, but do not connect to any vertices in the original graph other than $s_1$, $s_2$, and $s_3$. Thus, there are two modes of traversing from $s_1$ to $s_3$ via a simple path: (1) traverse the new path from $s_1$ to $s_3$ denoted as $p^*$ by Algorithm~\ref{alg:problemMap}, or (2) move from $s_1$ to $s_2$ via edges from $\Enew$, then from $s_2$ to $s_3$ via edges from $\Enew$. By construction, the path directly from $s_1$ to $s_3$ (the path added in the final loop of Algorithm~\ref{alg:problemMap}) passes through $2N-2$ intermediate nodes, having a length of $2N-1$. Taking the indirect route first requires taking one of the $M+1$ paths from $s_1$ to $s_2$, which has length $N$, then taking one of the $M+1$ paths from $s_2$ to $s_3$, which also has length $N$. Thus, the total length of any path via $s_2$ is $2N$, which is longer than $p^*$. Thus, $p^*$ is the shortest path from $s_1$ to $s_3$ in $\hat{G}^\prime=(\hat{V}, \hat{E}\setminus E^\prime)$, so $E^\prime$ is a solution to $f(a)$.
\end{proof}
An immediate consequence of Lemma~\ref{lem:3TCSolvesFPC} is that condition (\ref{item:scaleOpt}) is satisfied, as stated formally below.
\begin{corollary}
If $\instA$ is an instance of 3-Terminal Cut, then $\opt(f(\instA))\leq\opt(\instA)$, satisfying condition (\ref{item:scaleOpt}) with $\alpha=1$. 
\label{cor:optimal}
\end{corollary}
\begin{proof}
Let $E^\prime$ be the optimal solution to $\instA$, i.e., $|E^\prime|=\opt(\instA)$. By Lemma~\ref{lem:3TCSolvesFPC}, $E^\prime$ also solves $f(\instA)$. Thus, the optimal solution to $f(\instA)$ can be no larger than $|E^\prime|$, and therefore $\opt(f(\instA))\leq\opt(\instA)$.
\end{proof}
While the above corollary is sufficient to satisfy condition (\ref{item:scaleOpt}), we can make a stronger statement that is useful to prove condition (\ref{item:scaleSol}): the optimal solutions of the two problems are the same.
\begin{lemma}
For an instance of 3-Terminal Cut $\instA$, $\opt(\instA)=\opt(f(\instA))$. In particular, if $E^\prime$ is an optimal solution for $\instA$, then it is also an optimal solution for $f(\instA)$.
\label{lem:sameOpt}
\end{lemma}
\begin{proof}
Let $G=(V,E)$ be the graph in $\instA$, and $\hat{G}=(\hat{V}, \hat{E})$ be the graph in problem $f(\instA)$. Let $\hat{E}^\prime$ be an optimal solution to $f(\instA)$. Partition $\hat{E}^\prime$ into the edges that occur in the original graph $E_1=\hat{E}^\prime\cap E$, and those that do not, $E_2=\hat{E}^\prime\setminus E$. By Lemma~\ref{lem:3TCSolvesFPC}, if $E_1$ is a solution to $\instA$, it is also a solution to $f(\instA)$. Therefore, if $E_1$ is a solution to $\instA$, $E_2=\emptyset$ (otherwise we contradict the assumption that $\hat{E}^\prime$ is an optimal solution to $f(\instA)$. Thus, we focus on the case where $E_1$ is not a solution to $\instA$. In this case, within the graph $G_1=(V, E\setminus E_1)$, not all terminals $s_1$, $s_2$, and $s_3$ are disconnected. If there is a path from $s_1$ to $s_3$, the length of this path is at most $N-1$, which is shorter than $p^*$. This contradicts the assumption that $\hat{E}^\prime$ is a solution to $f(\instA)$. There are two other possibilities: $s_1$ and $s_2$ are connected in $G_1$, or $s_2$ and $s_3$ are connected. If $s_1$ and $s_2$ are connected, there is a path between the terminals of length at most $N-2$ (excluding $s_3$). Algorithm~\ref{alg:problemMap} inserts $M+1$ independent parallel paths from $s_2$ to $s_3$. If any of these paths remains, there is a path in $\hat{G}^\prime=(\hat{V}, \hat{E}\setminus\hat{E}^\prime)$ from $s_2$ to $s_3$ of length $N$, which would create a path from $s_1$ to $s_3$ of length at most $2N-2$, which is shorter than $p^*$. Thus, $\hat{E}^\prime$ would have to include at least one edge from all $M+1$ parallel paths from $s_2$ to $s_3$ inserted by Algorithm~\ref{alg:problemMap}. This means that $|E_2|\geq M+1$, implying that $|\hat{E}^\prime|\geq M+1$. This contradicts the assumption that $\hat{E}^\prime$ is an optimal solution: any solution to $\instA$ is a solution to $f(\instA)$ and its cost is at most $M$. The analogous argument holds if $s_2$ and $s_3$ are connected in $G_1$. Thus, a solution to $f(\instA)$ cannot be optimal if it does not include a solution to $\instA$, implying that $\opt(f(\instA))\geq\opt(\instA)$. This in conjunction with Corollary~\ref{cor:optimal} proves that the optima of $\instA$ and $f(\instA)$ are the same.
\end{proof}

To show that the reduction also satisfies condition (\ref{item:scaleSol}), we first prove the following lemma.
\begin{lemma}
Let $\instA$ be an instance of 3-Terminal Cut where all weights are equal, with $G=(V, E)$. Let $\instB=f(\instA)$ be the instance of Force Path Cut obtained by applying Algorithm~\ref{alg:problemMap} to $\instA$, with $\hat{G}=(\hat{V}, \hat{E})$. Further, let $\hat{E}^\prime$ be a solution to $\instB$. Then $g(\hat{E}^\prime)$ is a solution to $\instA$.
\label{lem:gSolves3TC}
\end{lemma}
\begin{proof}
The case where $|\hat{E}^\prime|\geq M$ is trivial: If all edges are removed, then the terminals are disconnected. With the assumption that $|\hat{E}^\prime| < M$, partition $\hat{E}^\prime$ into two parts: $E_1=\hat{E}^\prime\cap E$ and $E_2=\hat{E}^\prime\setminus E1$, i.e., $E_1$ is the edges in the solution to $\instB$ that existed in the original graph, and $E_2$ consists of the edges in the solution added by $f$ (Algorithm~\ref{alg:problemMap}). Since $\hat{E}^\prime$ is a solution to $\instB$, $p^*$ is the shortest path from $s_1$ to $s_3$ in the graph $\hat{G}^\prime=(\hat{V}, \hat{E}\setminus\hat{E}^\prime)$. The length of $p^*$ is $2N-1$. Assume $E_1$ is not a solution to $\instA$, i.e., that $s_1$, $s_2$, and $s_3$ are not disconnected in $G_1=(V, E\setminus E_1)$. Then there would be a path between at least two of the three terminals in $G_1$. If there were a path from $s_1$ to $s_3$, the length of this path would be at most $N-1<2N-1$, so this contradicts the assumption that $p^*$ is the shortest path in $\hat{G}^\prime$. For the other cases, assume there is no such path (i.e., $s_1$ and $s_3$ are disconnected in $G_1$). Suppose there is a path between $s_1$ and $s_2$. This path's length is at most $N-2$ (since it cannot include $s_3$). In addition, either (1) one of the paths from $s_2$ to $s_3$ added by $f$ remains, or (2) all $M+1$ such paths were cut, which contradicts our assumption since it would require $\hat{E}^\prime$ to contain at least $M+1$ edges. The paths from $s_2$ to $s_3$ added in Algorithm~\ref{alg:problemMap} have length $N$, thus creating a path from $s_1$ to $s_3$ with length at most $2N-2<2N-1$, which also contradicts the assumption that $p^*$ is the shortest path. A similar argument is made for the case where a path from $s_2$ to $s_3$ remains in $G_1$: Such a path would have length at most $N-2$, and at least one of the paths added between $s_1$ and $s_2$ remains, resulting in a path of length at most $2N-2$. Thus, in the case where $|\hat{E}^\prime|< M$, $\hat{E}^\prime\cap E$ is a solution to $\instA$.
\end{proof}

With this result, we can show the reduction meets the final criterion.
\begin{lemma}
If $\instA$ is an instance of 3-Terminal Cut with $G=(V, E)$, $y$ is a solution to $\instB=f(\instA)$, and $x=g(y)$, then $$|\cost(x)-\opt(\instA)|\leq |\cost(y)-\opt(\instB)|,$$
satisfying (\ref{item:scaleSol}) with $\beta=1$. 
\label{lem:solutionCostScales}
\end{lemma}
\begin{proof}
Let $\hat{E}^\prime$ be a solution to $\instB$, and partition $\hat{E}^\prime$ into $E_1=\hat{E}^\prime\cap E$ and $E_2=\hat{E}^\prime\setminus E_1$. From Lemma~\ref{lem:gSolves3TC}, we know that $g(\hat{E}^\prime)$ solves $\instA$. From Lemma~\ref{lem:sameOpt}, we also know the optimal solutions are the same size. If $|\hat{E}^\prime|\geq M$, then $g(\hat{E}^\prime)=E$, and we have
\begin{equation}
    \left||E|-\opt(\instA)\right|=|M-\opt(\instA)|=|M-\opt(\instB)|\leq\left||\hat{E}^\prime|-\opt(\instB)\right|,
\end{equation}
so the condition holds. If $|\hat{E}^\prime|< M$, we have
\begin{equation}
    \left||\hat{E}^\prime\cap E|-\opt(\instA)\right|=\left||\hat{E}^\prime\cap E|-\opt(\instB)\right|\leq\left||\hat{E}^\prime|-\opt(\instB)\right|,
\end{equation}
which proves the claim.
\end{proof}
These intermediate results show that the proposed reduction is a linear reduction of 3-Terminal Cut to Force Path Cut, implying that Force Path Cut is APX-hard.
\begin{proof}[Proof of Lemma~\ref{lem:undirected}]
By construction, $f$, as described by Algorithm~\ref{alg:problemMap}, takes an instance of 3-Terminal Cut and maps it to an instance of Force Path Cut. As shown in the pseudocode, this takes $O(MN)$ time to compute. By Lemma~\ref{lem:gSolves3TC}, $g$ maps a solution to the instance of Force Path Cut obtained via $f$ to a solution to the original 3-Terminal Cut problem. The procedure of removing the original edge set takes polynomial time that varies depending on the data structure, e.g., $O(MN)$ per removal in an adjacency list. Thus, $f$ and $g$ are appropriate mappings that take polynomial time to compute, satisfying condition (\ref{item:polynomial}). By Corollary~\ref{cor:optimal}, the reduction satisfies condition (\ref{item:scaleOpt}), and by Lemma~\ref{lem:solutionCostScales}, it satisfies (\ref{item:scaleSol}). This means that $f$ and $g$ provide a linear reduction from 3-Terminal Cut to Force Path Cut. Since 3-Terminal Cut is APX-hard, Force Path Cut for undirected graphs is APX-hard as well.
\end{proof}
\subsection{Proof for Directed Graphs}
\label{subsec:proofDirected}
To prove that Force Path Cut is APX-hard for directed graphs, we formulate a linear reduction from Force Path Cut for undirected graphs to the directed case. The linear reduction in this case is simple. The function $f$ that maps an undirected instance of Force Path Cut to a directed one simply takes each edge from the former and includes the two edges between the associated nodes in both directions. The values of $p^*$, $s$, and $t$ remain the same. Formally, $f$ replaces $E$ with
\begin{equation}
    \hat{E}=\bigcup_{\{u,v\}\in E}{\{(u, v), (v, u)\}}.
\end{equation}
The graph $\hat{G}=(V, \hat{E})$ can be constructed in $O(N+M)$ time.

If we have a solution to Force Path Cut on the directed graph $\hat{G}$, there is a similarly simple mapping to a solution to undirected Force Path Cut: include the undirected version of each edge in the solution. This takes $O(M)$ time. We show that this mapping provides a solution to the original problem in the following lemma.
\begin{lemma}
Let $\instA$ be an undirected instance of Force Path Cut, and $f(\instA)$ be its corresponding directed instance. If $\hat{E}^\prime$ is a solution to $f(\instA)$, then the undirected solution
\begin{equation}
    E^\prime=\bigcup_{(u,v)\in \hat{E}^\prime}{\{\{u, v\}\}}
\end{equation}
solves $\instA$.\label{lem:directedSolvesUndirected}
\end{lemma}
\begin{proof}
Suppose $E^\prime$ did not solve $\instA$, i.e., $p^*$ is not the shortest path from $s$ to $t$ in $G^\prime=(V, E\setminus E^\prime)$. Then there is some other path, $\hat{p}$, from $s$ to $t$ that is not longer than $p^*$. However, this path also exists in $\hat{G}^\prime=(V, \hat{E}\setminus\hat{E}^\prime)$: all edges from $E$ were added in the creation of $\hat{E}$, so if $\hat{p}$ were not in $\hat{G}^\prime$, at least one edge from $\hat{p}$ would have to be in $\hat{E}^\prime$. The mapping $g$ would include the undirected version of this edge in $E^\prime$, which would cut $\hat{p}$ in $G$ as well. Thus, the existence of $\hat{p}$ contradicts the assumption that $\hat{E}^\prime$ is a solution to $f(\instA)$, proving the claim.
\end{proof}
\begin{lemma}
Let $E^\prime$ be a solution to Force Path Cut on a directed graph. If $(u,v)\in E^\prime$ and $(v, u)\in E^\prime$, then either $E^\prime\setminus\{(u,v)\}$ or $E^\prime\setminus\{(v,u)\}$ is also a solution to Force Path Cut.\label{lem:cutTwoWay}
\end{lemma}
\begin{proof}
Let $d(v_1, v_2)$ be the distance between $v_1$ and $v_2$ in the directed graph $G^\prime=(V, E\setminus E^\prime)$. (If a path does not exist between $v_1$ and $v_2$, then $d(v_1,v_2)=\infty$.)

First, consider the case where both $u$ and $v$ are on $p^*$. Assume without loss of generality that $u$ precedes $v$ on the path. Then $(v,u)$ did not need to be removed from the edge set. Since $p^*$ is the shortest path from $s$ to $t$, it is comprised of shortest paths between all intermediate nodes on the path; for example, the shortest path from $s$ to $v$ and the shortest path from $v$ to $t$. Removing $(v,u)$ from $E^\prime$---leaving the edge in the graph when solving Force Path Cut---would not change the status of $p^*$ as the shortest path from $s$ to $t$: The shortest path  from $s$ to $v$ would still include $u$, and adding $(v,u)$ back to the graph would only enable moving backward along the path. This means that $E^\prime\setminus\{(v,u)\}$ is also a solution.

Now consider a case where one of $u$ and $v$ is part of $p^*$, but the other is not. Without loss of generality, assume $u$ is on $p^*$ and $v$ is not. Let $\ell_{p^*}$ be the length of $p^*$. Since $p^*$ is the shortest path, we know that $d(s, v)+d(v, t)> \ell_{p^*}$. Since $u$ is on $p^*$, we have $d(s, u)+d(u,t)=\ell_{p^*}$. Suppose the claim does not hold: that neither $E^\prime\setminus\{(u,v)\}$ nor $E^\prime\setminus\{(v,u)\}$ is a solution. This means that adding $(u,v)$ back into the edge set must create a path not longer than $p^*$, so we have $d(s, u)+d(v, t)+1\leq \ell_{p^*}$. The same is true for $(v, u)$, implying that $d(s, v)+d(u, t)+1\leq\ell_{p^*}$. Adding the latter two inequalities, we have 
\begin{equation}
    d(s,u)+d(u,t)+d(s,v)+d(v,t)+2\leq2\ell_{p^*}\Rightarrow d(s,v)+d(v,t)\leq\ell_{p^*}-2,\label{eq:pathContradiction}
\end{equation}
where we use the equation $d(s, u)+d(u,t)=\ell_{p^*}$. This contradicts the first inequality, that the path from $s$ to $t$ via $v$ is longer than $p^*$. In the case where neither $u$ nor $v$ is on $p^*$, we replace $d(s, u)+d(u, t)=\ell_{p^*}$ with $d(s, u)+d(u, t)>\ell_{p^*}$, and the inequality in (\ref{eq:pathContradiction}) becomes strict. Thus, if both $(u,v)$ and $(v,u)$ are in $E^\prime$, then either $E^\prime\setminus\{(u,v)\}$ or $E^\prime\setminus\{(v,u)\}$ is also a solution to Force Path Cut.
\end{proof}

This lemma has an immediate consequence that is important for proving the reduction is linear.
\begin{corollary}
Let $E^\prime$ be an optimal solution to Force Path Cut on a directed graph. If $(u,v)\in E^\prime$, then $(v, u)\notin E^\prime$.
\end{corollary}
\begin{proof}
From Lemma~\ref{lem:cutTwoWay}, if $E^\prime$ contained both $(u,v)$ and $(v,u)$, then removal of one of these edges would still be a solution. Since the resulting solution would be smaller than $E^\prime$, this contradicts the premise of the claim.
\end{proof}
In addition, we obtain the optimal solution to the undirected problem if we find the optimal solution to the directed problem via the reduction.
\begin{lemma}
Let $\instA$ be an undirected instance of Force Path Cut and $f(\instA)$ be the corresponding directed instance. Then the optimal solution to $f(\instA)$ is the optimal solution to $\instA$.
\label{lem:optEqual}
\end{lemma}
\begin{proof}
By Lemma~\ref{lem:directedSolvesUndirected}, we know that the optimal solution to $f(\instA)$ solves $\instA$, so $\opt(\instA)\leq\opt(f(\instA))$. Let $E^\prime$ be an optimal solution to $\instA$. Let $\hat{E}^\prime$ be a solution to $f(\instA)$ that includes both directed edges for each undirected edge in $E^\prime$, i.e.,
\begin{equation}
    \hat{E}^\prime=\bigcup_{\{u,v\}\in E^\prime}{\{(u,v), (v,u)\}}.
\end{equation}
Since $E^\prime$ is a solution to $\instA$, $\hat{E}^\prime$ is a solution to $f(\instA)$. (Otherwise a path $\hat{p}$ in $f(\instA)$ that is shorter than $p^*$ would not be removed by $E^\prime$.) However, $\hat{E}^\prime$ contains a pair of edges for each edge in $E^\prime$: for any $(u,v)\in\hat{E}^\prime$, we have $(v,u)\in\hat{E}^\prime$. By Lemma~\ref{lem:cutTwoWay}, one edge from each pair can be removed. This means that, for an optimal solution to $\instA$, we can find a solution to $f(\instA)$ that is the same size, which implies that $\opt(f(\instA))\leq\opt(\instA)$. This means that $\opt(\instA)=\opt(f(\instA))$, and an optimal solution for one problem can be applied to the other.
\end{proof}

Combining these intermediate results, we show that the reduction from undirected Force Path Cut to the directed version is linear, and the directed version of the problem is also APX-hard.
\begin{lemma}
Force Path Cut is APX-hard for directed graphs, including the case where all weights and all costs are equal.
\label{lem:directed}
\end{lemma}
\begin{proof}
The function $f$ simply takes edges from an undirected graph and builds a directed graph with the same (directed) edges, which takes $O(N+M)$ time. The function $g$ takes a set of directed edges and converts it into a set of undirected edges, which takes $O(M)$ time. By Lemma~\ref{lem:directedSolvesUndirected}, $g$ maps to a true solution to the undirected Force Path Cut problem. This means that condition (\ref{item:polynomial}) is satisfied.

Lemma~\ref{lem:optEqual} guarantees that condition (\ref{item:scaleOpt}) is satisfied as well, with $\alpha=1$. Finally, let $\instA$ be an undirected instance of Force Path Cut. For any solution $\hat{E}^\prime$ to $f(\instA)$, we know that
\begin{equation}
    \frac{1}{2}|\hat{E}^\prime|\leq|g(\hat{E}^\prime)|\leq |\hat{E}^\prime|.
\end{equation}
 Thus, we have
 \begin{equation}
     \left||g(\hat{E}^\prime)|-\opt(\instA)\right|\leq\left||\hat{E}^\prime|-\opt(\instA)\right|=\left||\hat{E}^\prime|-\opt(\instB)\right|,
 \end{equation}
 so condition (\ref{item:scaleSol}) is satisfied with $\beta=1$.
 
 Since all three conditions are satisfied, $f$ and $g$ provide a linear reduction from Force Path Cut on an undirected graph to the same problem on a directed graph. By Lemma~\ref{lem:undirected}, Force Path Cut for undirected graphs is APX-hard, so the reduction implies it is APX-hard for directed graphs as well.
\end{proof}
Theorem~\ref{thm:pathHard} is a direct consequence of Lemma~\ref{lem:undirected} and Lemma~\ref{lem:directed}.
\section{PATHATTACK Convergence}
\label{sec:convergence}
We consider the case where the ellipsoid algorithm is used to optimize the relaxed version of the integer program. At each iteration, we consider the center of an ellipsoid and determine whether this point violates any constraints. Call this point $\Delta_f\in[0,1]^M$. In addition, let $P$ be the current set of explicit path constraints and $P_f$ be the set of path constraints---both implicit and explicit---that $\Delta_f$ does not violate.

Given $\Delta_f$, we perform the randomized rounding routine used in Algorithm~\ref{alg:LPCut}. With probability at least $1/2$, this procedure will yield a result that is within the guaranteed approximation margin (i.e., the objective of the integer solution is within $\ln{(4|P|)}$ of the fractional solution) and satisfying all explicit constraints. Thus, with high probability, it will yield such a solution in $O(\log{|E|})$ trials. Note that this holds for the full set $P_f$ if $|P_f|\leq \frac{1}{4}e^{\lceil\ln{(4|P|)}\rceil}$, since the two set sizes result in the same success bounds. If we attempt for $O(\log{|E|})$ trials and never find a violated constraint, we increment the number of Bernoulli random variables used in the randomized rounding procedure and the approximation factor. With high probability, this will yield a valid solution if $|P_f|\leq e^{\lceil\ln{(4|P|)}\rceil+1}$. We continue until we find a path that needs to be cut that is not (fractionally) cut by the solution of the relaxed problem. Algorithm~\ref{alg:oracle} provides pseudocode for this procedure. The algorithm will complete in polynomial time: There are at most $N\ln{N}$ iterations of the outer loop, at most $\lceil\log_2{|E|}\rceil$ iterations of the second loop, inside of which:
\begin{itemize}
    \item At most $O(|E|N\log{N})$ Bernoulli random variables instantiated and aggregated
    \item An $O(|E|)$-length vector is created
    \item Edges are removed from a graph (at most $O(|E|)$ time)
    \item In the worst case, the \emph{second} shortest path is found, which takes $O(N^3)$ time
    \item The constraints are checked ($O(N)$ time).
\end{itemize}
Each item takes polynomial time to complete, so the overall algorithm takes polynomial time.
\begin{algorithm}
\KwIn{Graph $G=(V, E)$, weights $w$, costs $\cVec$, path $p^*$, fractional cut vector $\Delta_f$, path set $P$}
\KwOut{Path $p$ that violates}
$s\gets$ first node of $p^*$\;
$t\gets$ last node of $p^*$\;
$\approxFactor\gets\lceil\ln{(4|P|)}\rceil$\;
$\hat{\Delta} \gets \mathbf{0}$\;
$E^\prime\gets\emptyset$\;
not\_cut$\gets\mathbf{True}$\;
\While{$(\approxFactor<N\ln{N})$ $\mathbf{and}$ $\mathrm{not\_cut}$}{
$\ctr\gets0$\;
\While{$(\cVec^\top\Delta_f > \cVec^\top\hat{\Delta}(4\cdot\approxFactor)$ $\mathbf{or}$ $\mathrm{not\_cut})$ $\mathbf{and}$ $\ctr<\lceil\log_2{|E|}\rceil$}{
$E^\prime\gets\emptyset$\;
\For{$i\gets 1$ to $\approxFactor$}{
\tcp{randomly select edges based on $\Delta_f$}
$E_1\gets\{e\in E \textrm{ with probability }\Delta_f[e]\}$\;
$E^\prime\gets E^\prime\cup E_1$\;
}
$\hat{\Delta}\gets$ indicator vector for $E^\prime$\;
$G^\prime\gets (V, E\setminus E^\prime)$\;
$p\gets$ shortest path (other than $p^*$) from $s$ to $t$ in $G^\prime$ with weights $w$\;
\eIf{$\sum_{(u,v)\in p}{\Delta_f[u, v]} \geq 1$}{
\tcp{$p$ does not violate $\Delta_f$}
not\_cut$\gets$\textbf{True}\;
}{
not\_cut$\gets$\textbf{False}\;
\If{$p$ is longer than $p^*$}{
\tcp{found solution (if within approximation factor)}
$p\gets\emptyset$\;
}
}
$\ctr\gets\ctr+1$\;
}
$\approxFactor\gets\approxFactor+1$
}
\Return $p$
\caption{Randomized Rounding Oracle}
\label{alg:oracle}
\end{algorithm}

\section{PATHATTACK for Node Removal}
\label{sec:nodeRemoval}
As discussed in Section~\ref{subsec:nodeRemoval}, \PATHATTACK{} can be used for node removal in addition to edge removal. We refer to the associated problem as \emph{Force Path Remove}: Given a weighted graph $G=(V,E)$, $w:E\rightarrow\reals_{\geq0}$, where each node has a cost of removal, $c:V\rightarrow\reals_{\geq0}$, a path $p^*$ from $s\in V$ to $t\in V$, and a budget $b$, is there a subset of nodes $V^\prime\subset V$ such that $\sum_{v\in V^\prime}{c(v)}\leq b$ and $p^*$ is the shortest path from $s$ to $t$ in $G^\prime=(V\setminus V^\prime, E\setminus E_{V^\prime})$, where $E_{V^\prime}\subset E$ is the set of edges with at least one node in $V^\prime$? In this section, we prove that this problem is also APX-hard, and provide empirical results on the same datasets as shown in Section~\ref{sec:experiments}.
\subsection{Computational Complexity}
We propose the following reduction from Force Path Cut to Force Path Remove for unweighted, undirected graphs:
\begin{itemize}
    \item Take the input graph $G=(V,E)$ and create its line graph $\hat{G}=(\hat{V},\hat{E})$
    \item Create new nodes $\hat{s}$ and $\hat{t}$ and add these to $\hat{V}$
    \item For all edges $e\in E$ incident to $s$, create a new edge in $\hat{E}$ connecting $\hat{s}$ to the corresponding node $v_e\in\hat{V}$
    \item For all edges $e\in E$ incident to $t$, create a new edge in $\hat{E}$ connecting $\hat{t}$ to the corresponding node $v_e\in\hat{V}$
    \item Let $\hat{p}^*$ be the path from $\hat{s}$ to $\hat{t}$ where all intermediate nodes are those that correspond to the edges in $p^*$, in the same order
    \item Solve Force Path Remove on $\hat{G}$ targeting $\hat{p}^*$.
\end{itemize}
This procedure is the function $f$ that maps an instance of Force Path Cut to an instance of Force Path Remove. When we solve Force Path Remove, we obtain a set of nodes $\hat{V}^\prime$ to remove. Since all nodes in $\hat{G}$ other than $\hat{s}$ and $\hat{t}$ correspond to edges in $E$---and $\hat{s}$ and $\hat{t}$ cannot be in $\hat{V}^\prime$---the function $g$ mapping a solution to Force Path Remove to a solution to Force Path Cut is to make $E^\prime$ the set of edges in $E$ corresponding to the nodes in $\hat{V}^\prime$. 

We first prove that the reduction yields a solution to Force Path Cut in the following lemma.
\begin{lemma}
Let $\instA$ be an unweighted, undirected instance of Force Path Cut and $\hat{V}^\prime$ be a solution to $f(\instA)$. Then $E^\prime=g(\hat{V}^\prime)$ is a solution to $\instA$.
\end{lemma}
\begin{proof}
Suppose $E^\prime$ were not a solution to $\instA$. This means that $p^*$ is not the shortest path from $s$ to $t$ in $G^\prime=(V, E\setminus E^\prime)$. Let $p^\prime$ be the competing path. Since $p^\prime$ remains a path from $s$ to $t$, none of its edges are in $E^\prime$, and therefore none of those edges' corresponding nodes from $\hat{G}$ are in $\hat{V}^\prime$. This means that all edges in $p^\prime$ have nodes in $\hat{G}$. These nodes, however, form a path $\hat{p}^\prime$ from $\hat{s}$ to $\hat{t}$, and since the length of $p^\prime$ is less than or equal to the length of $p^*$, $\hat{p}^\prime$ is also no longer than $\hat{p}^*$, which contradicts the assumption that removing $\hat{V}^\prime$ makes $\hat{p}^*$ the shortest path from $\hat{s}$ to $\hat{t}$.
\end{proof}

In addition, the optimal solution to Force Path Cut is also the optimal solution to the derived Force Path Remove instance.
\begin{lemma}
Let $\instA$ be an instance of Force Path Cut. If $\hat{V}^\prime$ is an optimal solution of $f(\instA)$, then $E^\prime=g(\hat{V}^\prime)$ is an optimal solution of $\instA$.
\label{lem:optNodeRemovalIsOptEdgeRemoval}
\end{lemma}
\begin{proof}
Suppose there were a better solution $E_1\subset E$, $|E_1|<|E^\prime|$. Consider the corresponding set of nodes $V_1\subset\hat{V}$. The augmented line graph of $G_1=(V, E\setminus E_1)$ would be the same as $\hat{G}$ with $V_1$ (and all edges adjacent to nodes in $V_1$) removed. Since $p^*$ is the shortest path from $s$ to $t$ in $G_1$, however, $\hat{p}^*$ would be the shortest path from $\hat{s}$ to $\hat{t}$ in the augmented line graph. This means that $V_1$ is a solution to $f(\instA)$. The fact that $|E_1|<|E^\prime|$, however, implies that $|V_1|<|\hat{V}^\prime|$, which contradicts the assumption that $\hat{V}^\prime$ is an optimal solution, proving the claim.
\end{proof}
Using these results, we now prove APX-hardness in the undirected case.
\begin{lemma}
Force Path Remove is APX-hard for undirected graphs, including the case where all costs are equal.\label{lem:removeUndirected}
\end{lemma}
\begin{proof}
To prove the claim, we show that the reduction described above satisfies the requirements of a linear reduction: $f$ and $g$ can be computed in polynomial time, the optimal solutions differ by at most a constant factor, and the absolute difference of any solution from the optimal solution. We consider each criterion in turn.

To create the line graph, a new graph is created with $|E|$ nodes. For each node in the original graph, edges are created in the new graph connecting two of the incident edges. Thus, for a node in $V$ with degree $d$, we have $\binom{d}{2}$ edges in $\hat{E}$, meaning that $|\hat{E}|$ is $O(|E|^2)$. After creating the line graph, we add two new nodes and their associated edges, which will be a total of $O(|V|)$. Overall, the new graph has $|E|+2$ nodes and $O(|E|^2+|V|)$ edges. Finally, to identify $\hat{p}^*$, we only need to keep track of which nodes are associated with edges in $p^*$ in the original graph, which requires $O(|V|)$ additional time to convert the node labels. This means that $f$ can be computed in polynomial time. To compute $g$, we only need to maintain a mapping of the edges in $G$ to the nodes in $\hat{G}$, which will take $O(|E|)$ time to populate. Converting the solution to Force Path Remove to a solution to Force Path Cut will take $O(|V^\prime|)=O(|E|)$ lookups in the mapping. Thus, both $f$ and $g$ can be computed in polynomial time.

By Lemma~\ref{lem:optNodeRemovalIsOptEdgeRemoval}, if $V^\prime$ is the optimal solution for $f(\instA)$, then $g(V^\prime)$ is the optimal solution for $\instA$, which means condition~(\ref{item:scaleOpt}) is satisfied with $\alpha=1$. In addition, since the solutions of the two problems are the same size---each node removed to solve $f(\instA)$ corresponds with an edge removed to solve $\instA$--- we have, for any proposed solution $x$ to $f(\instA)$,
\begin{equation}
    |\cost(x)-\opt(f(\instA))|=|\cost(g(x))-\opt(\instA)|,
\end{equation}
satisfying condition (\ref{item:scaleSol}) with $\beta=1$. Thus, all conditions of a linear reduction are satisfied. Since Force Path Cut is APX-hard for unweighted, undirected graphs (Lemma~\ref{lem:undirected}), this implies that Force Path Remove is APX-hard as well.
\end{proof}

We can trivially reduce from the undirected case to the directed case by creating a directed edge set that includes edges in both directions for each undirected edge in the original graph. Solving Force Path Remove on the resulting directed graph yields the same solution as solving the Force Path Remove on the original graph, thus corresponding solutions in the directed and undirected cases are always equally costly. The formal proof of the following lemma is straightforward, and we omit it for brevity.
\begin{lemma}
Force Path Remove is APX-hard for directed graphs, including the case where all costs are equal.\label{lem:removeDirected}
\end{lemma}

The following theorem is a direct consequence of Lemmas~\ref{lem:removeUndirected} and \ref{lem:removeDirected}.
\begin{theorem}
Force Path Remove is APX-hard, including the case where all costs are equal.
\end{theorem}

\subsection{Experiments}
We use the same experimental setup as in Section~\ref{subsec:setup}, in this case only considering the 200th shortest path between the source and target. The baseline method is analogous to the edge removal case: the lowest-cost node along the shortest path is removed until $p^*$ is the shortest path. 
In all cases, the cost of edge removal is the edge weight, while the cost of node removal is degree. For node removal, we only consider values of $p^*$ that have viable solutions (i.e., $p^*$ is the shortest path from $s$ to $t$ in the subgraph induced on the nodes used by $p^*$).

Figure~\ref{fig:plots_unweighted_NR} shows results on unweighted graphs. We consider the equal-weighted case, without random weights added, as this case highlights the greatest difference between edge cuts and node removal. We did not consider cliques, since there is never a viable solution when $p^*$ is more than one hop. Results for edge removal are shown for comparison. There are a few substantial differences between relative performance using edge removal and node removal. For example, \PATHATTACK{} provides a much more significant cost reduction with node removal rather than edge removal for Erd\H{o}s--R\'{e}nyi graphs. In the vast majority of these cases ($99/100$ for edge removal, $97/100$ for node removal), randomized \PATHATTACK{} finds the optimal solution. This suggests that targeting the low-degree nodes---as the baseline does for node removal---is not an effective strategy in this context, i.e., the increased cost of removing higher-degree nodes pays off by removing more competing paths. Another major difference is that the baseline outperforms \PATHATTACK{} for lattices with equal weights. In this case, while \PATHATTACK{} finds the optimal solution with edge removal in $94\%$ of cases, it is only optimal $20\%$ of the time with node removal. Since degrees are equal, \texttt{GreedyCost} simply removes the first node that deviates from $p^*$, which works well for unweighted lattices, while the analogous strategy for edge removal is suboptimal. We see a similar phenomenon with the autonomous system graph: the baseline method of removing low-degree nodes yields a near-optimal solution, suggesting that focusing on disrupting the intermediate low-degree nodes between the endpoints and the hubs is a cost-effective strategy. While attacking high-degree nodes is useful for disconnecting networks with skewed degree distributions~\cite{Albert2000}, the degree-based cost and focus on making a particular path shortest make a different strategy ideal here.
\begin{figure}[t]
    \centering
    \includegraphics[width=\textwidth]{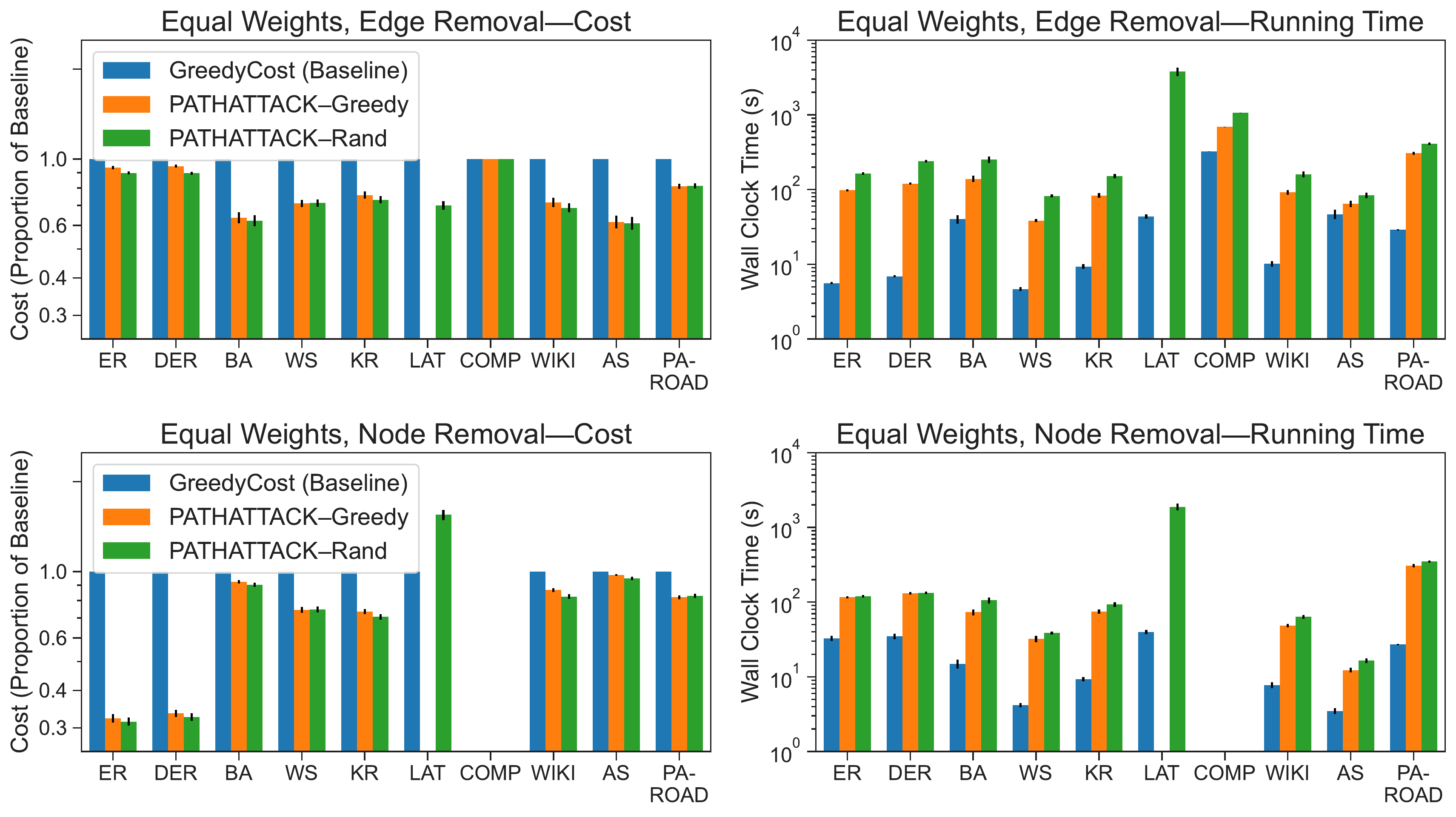}
    \caption{Results on unweighted networks when \PATHATTACK{} removes nodes rather than edges. Results are shown where all edges have equal weight and removal cost, in terms of attack cost (left column) and running time (right column). Lower is better for both metrics. Bar heights are means across these trials and error bars are standard errors. Results are shown for edge removal (top row) as a comparison to the results for node removal (bottom row). For (D)ER graphs, \PATHATTACK{} yields much more substantial improvement with node removal than edge removal, while for lattices \PATHATTACK{} actually underperforms the baseline. This demonstrates the variation in effectiveness of the baseline (greedily removing low-degree nodes) over differences in topology.}
    \label{fig:plots_unweighted_NR}
\end{figure}

Figure~\ref{fig:plots_weighted_NR} illustrates results on real weighted graphs. Unlike the autonomous system graph, here the power grid and road networks outperform the baseline more significantly using node removal. Here, the diversity of edge weights among the real graphs---and their contribution to path length---makes greedily selecting low-cost edges a more effective heuristic than selecting low-degree nodes.
\begin{figure}[t]
    \centering
    \includegraphics[width=\textwidth]{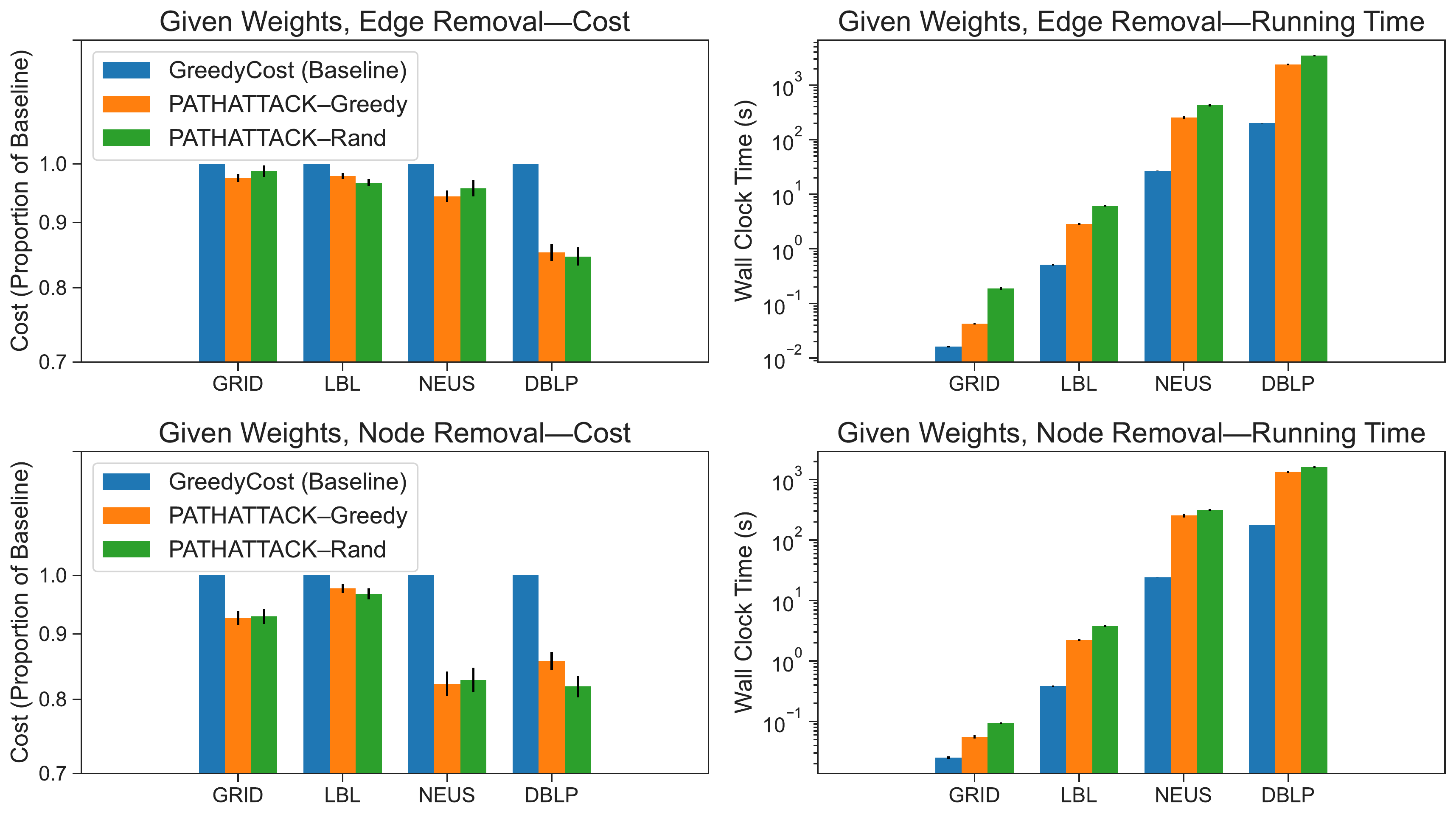}
    \caption{Results on weighted networks when \PATHATTACK{} removes nodes rather than edges. Results are shown in terms of attack cost (left column) and running time (right column). Lower is better for both metrics. Bar heights are means across these trials and error bars are standard errors. Results are shown for edge removal (top row) as a comparison to the results for node removal (bottom row). Unlike the unweighted graphs, \PATHATTACK{} outperforms the baseline by a higher margin with node removal than with edge removal.}
    \label{fig:plots_weighted_NR}
\end{figure}

\section{Dataset Features}
\label{sec:datasets}
Our experiments were run on several synthetic and real networks across different edge-weight distributions. All networks are undirected except DER, WIKI, and LBL. We described the edge-weight distributions in Section~\ref{sec:experiments} of the paper. 
Table \ref{table:syn_net_prop} provides summary statistics of the synthetic networks.

\begin{table}[!ht]
\renewcommand{\arraystretch}{1.0}
\centering
\begin{tabular}{|l|c|c|c|c|c|c|c|c|} 
\hline
Networks & Nodes & Edges & $\langle k\rangle$ & $\sigma_k$ & $\kappa$ & $\tau$ & $\triangle$ & $\varphi$ \\
\hline \hline
ER  & $16,000$ & $159,880$ & $19.985$ & $4.469$ & $0.001$ & $0.001$ & $1,326$ & $1$ \\
 & $\pm 0$ & $\pm38$ & $\pm0.05$ & $\pm0.02$ & $\pm0.0$ & $\pm0.0$ & $\pm39$ & $\pm0$ \\
\hline
DER  & 16,000 & 319,901 & $39.988$ & $6.317$ & 0.001 & 0.001 & 10,639 & $1$\\
 & $\pm 0$ & $\pm 568$ & $\pm0.071$ & $\pm0.034$ & $\pm0.0$ & $\pm0.0$ & $\pm105$ & $\pm0 $ \\
\hline
BA  & $16,000$ & $159,900$ & $19.987$ & $24.475$ & $0.007$ & $0.006$ & $17,133$ & $1$ \\
 & $\pm0$ & $\pm0$ & $\pm0$ & $\pm0.3$ & $\pm0.0$ & $\pm0$ & $\pm500$ & $\pm0$ \\
\hline
WS  & 16,000 & 160,000 & 20 & 0.629 & 0.669 & 0.668 & 677682.42 & 1 \\
 & $\pm 0$ & $\pm 0$ & $\pm 0$ & $\pm 0.007$ & $\pm 0.001$ & $\pm 0.001$ & $\pm818.171$ & $\pm0$ \\
\hline
KR  & 16,337 & 159,595 & ~19.537~ & ~16.537~ & ~$0.003$~ & $0.005$ & $8,492$ & ~$1.18$~ \\
& $\pm22$ & $\pm94$ & $\pm0.02$ & $\pm1.32$ & $\pm0$ & ~$\pm0.002$~ & $\pm2,234$ & ~$\pm0.38$~ \\ 
\hline
LAT  & ~81,225~ & $161,880$ & $3.985$ & $0.118$ & $0$ & $0$ & $0$ & $1$ \\
 & $\pm0$& $\pm0$& $\pm0$& $\pm0$& $\pm0$& $\pm0$& $\pm0$& $\pm0$ \\ 
\hline
COMP~  & $565$ & ~159,330~ & $564$ & $0$ & $1$ & $1$ & ~$29,900,930$~ & $1$ \\
& $\pm0$& $\pm0$& $\pm0$& $\pm0$& $\pm0$& $\pm0$& $\pm0$& $\pm0$ \\
\hline
\end{tabular}
\caption{Properties of the synthetic networks used in our experiments. For each random graph model, we generate 100 networks. Note that the number of edges across the different networks is $\approx160$K. The table shows the average degree ($\langle k\rangle$), standard deviation of the degree ($\sigma_k$), average clustering coefficient ($\kappa$), transitivity ($\tau$), number of triangles ($\triangle$), and the number of components ($\varphi$). The $\pm$ values show the standard deviation across 100 runs of each random graph model.}
\label{table:syn_net_prop}
\end{table}

We ran experiments on both weighted and unweighted real networks. In cases of unweighted networks, we added edge weights from the same distributions as the synthetic networks. Below is a brief description of each network used. Table~\ref{table:net_prop} summarizes the properties of each network. 

The unweighted networks are:
\begin{itemize}
    \item Wikispeedia graph (WIKI): The network consists of Web pages (nodes) and connections (edges) created from the user-generated paths in the Wikispeedia game~\cite{West2009}. Available at \url{https://snap.stanford.edu/data/wikispeedia.html}.
    \item Oregon autonomous system network (AS): Nodes represent autonomous systems of routers and edges denote communication between the systems~\cite{Leskovec2005}. The dataset was collected at the University of Oregon on 31 March 2001. Available at \url{https://snap.stanford.edu/data/Oregon-1.html}.
    \item Pennsylvania road network (PA-ROAD): Nodes are intersections in Pennsylvania, connected by edges representing roads~\cite{Leskovec2009}. Available at \url{https://snap.stanford.edu/data/roadNet-PA.html}.
   
\end{itemize}

\begin{table}[ht]
\centering
\begin{tabular}{|l|c|c|c|c|c|c|c|c|} 
\hline
Networks & Nodes & Edges &  $\langle k\rangle$ & $\sigma_k$ & $\kappa$ & $\tau$ & $\triangle$ & $\varphi$ \\
\hline \hline
GRID & 347 & 444 & 2.559 & 1.967 & 0.086 & 0.087 & 40 & 1 \\ 
\hline
LBL & 3,186 & 15,553 & 9.763 & 40.702 & 0.048 & 0.001 & 1821 & 10 \\
\hline
WIKI & 4,592 & 119,882 & 52.213 & 78.601 & 0.195 & 0.158 & 550,545 & 2 \\
\hline
AS & 10,670 & 22,002 & 4.124 & 31.986 & 0.296 & 0.009 & 17,144 & 1 \\
\hline
PA-ROAD~ & ~1,088,092~ & ~1,541,898~ & 2.834 & 1.016 & 0.046 & 0.059 & 67,150 & 206 \\
\hline
NEUS & 1,524,453 & 1,934,010 & 2.537 & 0.950  & 0.022 & 0.030 & 37,012 & 1 \\
\hline
DBLP & 1,836,596 & 8,309,938 & 9.049 & ~21.381~ & ~0.631~ & ~0.165~ & ~26,912,200~ &  ~60,512~ \\ 
\hline
\end{tabular}
\caption{Properties of the real networks used in our experiments. For each network, we are listing the average degree ($\langle k\rangle$), standard deviation of the degree ($\sigma_k$), average clustering coefficient ($\kappa$), transitivity ($\tau$), number of triangles ($\triangle$), and the number of components ($\varphi$).}
\label{table:net_prop}
\end{table}

The weighted networks are:
\begin{itemize}
    \item Central Chilean Power Grid (GRID): Nodes represent power plants, substations, taps, and junctions in the Chilean power grid. Edges  represent transmission lines, with distances in kilometers~\cite{Kim2018}. The capacity of each line in kilovolts is also provided. Available at \url{https://figshare.com/collections/An_in-depth_network_structural_data_and_hourly_activity_on_the_Central_Chilean_power_grid/4053374}.
    \item Lawrence Berkeley National Laboratory network data (LBL): A graph of computer network traffic, which includes counts of the number of connections between machines over time. Counts are inverted for use as distances. Available at \url{https://www.icir.org/enterprise-tracing/download.html}.
    \item Northeast US Road Network (NEUS): Nodes are intersections in the northeastern part of the United States, interconnected by roads (edges), with weights corresponding to distance in kilometers. Available at \url{https://www.diag.uniroma1.it/\~challenge9/download.shtml}.
    \item DBLP coauthorship graph (DBLP): This is a co-authorship network~\cite{Benson2018}. We invert the number of co-authored papers to create a distance (rather than similarity) between the associated authors. Available at \url{https://www.cs.cornell.edu/\~arb/data/coauth-DBLP/}.
\end{itemize}

\end{document}